\documentclass[nofootinbib,aps,groupedaddress,longbibliography
]{revtex4-1}

\usepackage{bbm}
\usepackage{xcolor}
\usepackage{graphicx}
\usepackage{amssymb,latexsym}
\usepackage{tensor}
\usepackage{slantsc}
\usepackage{slashed}

\usepackage{mathtools}
\usepackage[normalem]{ulem}

\usepackage{marginnote}

\usepackage{url}

\usepackage{mathrsfs}
\usepackage{enumitem}

\newcommand{\Wpsi}{{\eta}}

\newcommand{\hdherr}{O\left( |h|^2_{\backg}\, |\bcov h|_{\backg}\right)}
\newcommand{\hdhsqerr}{O\left( |h|_{\backg}\, |\bcov h|^2_{\backg}\right)}
\newcommand{\herr}{O\left( |h|^3_{\backg}\right)}

\definecolor{applegreen}{rgb}{0.55, 0.71, 0.0}
\definecolor{armygreen}{rgb}{0.29, 0.33, 0.13}

\definecolor{caribbeangreen}{rgb}{0.0, 0.8, 0.6}

\newtheorem{thm}{Theorem}[section]
\newtheorem{theorem}[thm]{Theorem}
\newtheorem{lem}[thm]{Lemma}
\newtheorem{Def}{Definition}[section]
\newtheorem{prop}[thm]{Proposition}
\newtheorem{rem}[Def]{Remark}

\newcommand{\bcheckpsi}{{\check{\psi}}}

\newcommand{\sqrtbg}{{d\mu_{\backg}}}

\newcommand{\adsR}[1]{\overline{R}\tensor{\vphantom{R}}{#1}}

\definecolor{orange}{RGB}{255,127,0}

\newcommand{\quadratic}[1]{{#1}}

\newcommand{\gauge}[1]{{#1}}

\newcommand{\bcov}{\back{D}} 

\newcommand{\ptcheck}[1]{\ptc{checked on #1}}

\newcommand{\red}[1]{{\color{red} #1}}

\newcommand{\bR}{\back{R}} 

\newcommand{\Bgamma}{\backGamma{}}

\DeclareFontFamily{OT1}{rsfs}{}
\DeclareFontShape{OT1}{rsfs}{m}{n}{ <-7> rsfs5 <7-10> rsfs7 <10-> rsfs10}{}
\DeclareMathAlphabet{\mycal}{OT1}{rsfs}{m}{n}

{\catcode `\@=11 \global\let\AddToReset=\@addtoreset}
\AddToReset{equation}{section}

\newcounter{mnotecount}[section]

\newcommand{\T}{\mathbb{T}}

\newcommand{\back}[1]{\overline{#1}} 
\newcommand{\backg}{{{\overline{g}}}} 
\newcommand{\backGamma}{{{\overline{\Gamma}}}} 

\newcommand{\redad}[2]{\ptc{change or addition or rewording on #1}{\color{red} #2}}

\newcommand{\eel}[1]{\label{#1}\end{equation}}
\newcommand{\eeal}[1]{\label{#1}\end{eqnarray}}

\newcommand{\bel}[1]{\begin{equation}\label{#1}}
\newcommand{\bea}{\begin{eqnarray}}
\newcommand{\bean}{\begin{eqnarray}\nonumber}
\newcommand{\beal}[1]{\begin{eqnarray}\label{#1}}
\newcommand{\eea}{\end{eqnarray}}

\newcommand{\nn}{\nonumber}

\newcommand{\qed}{\hfill $\Box$}
\newcommand{\qedskip}{\hfill $\Box$\medskip}
\newcommand{\proof}{\noindent {\sc Proof:\ }}
\newcommand{\be}{\begin{equation}}
\newcommand{\eeq}{\end{equation}}
\newcommand{\ee}{\end{equation}}
\newcommand{\beqa}{\begin{eqnarray}}
\newcommand{\eeqa}{\end{eqnarray}}
\newcommand{\beqan}{\begin{eqnarray*}}
\newcommand{\eeqan}{\end{eqnarray*}}
\newcommand{\ba}{\begin{array}}
\newcommand{\ea}{\end{array}}

\newcommand{\mnote}[1]
{\protect{\stepcounter{mnotecount}}$^{\mbox{\footnotesize
$
\bullet$\themnotecount}}$ \marginpar{
\raggedright\tiny\em
$\!\!\!\!\!\!\,\bullet$\themnotecount: #1} }

\newcommand{\warn}[1]
{\protect{\stepcounter{mnotecount}}$^{\mbox{\footnotesize
$
\bullet$\themnotecount}}$ \marginpar{
\raggedright\tiny\em
$\!\!\!\!\!\!\,\bullet$\themnotecount: {\bf Warning:} #1} }

\newcommand{\R}{\mathbb R}

\newcommand{\eq}[1]{(\ref{#1})}

\newcommand{\ptc}[1]{\mnote{{\bf ptc:}#1}}

\newcommand{\beaa}{\begin{eqnarray*}}
\newcommand{\eeaa}{\end{eqnarray*}}

\def\ben{\begin{equation}}
\def\een{\end{equation}}
\def\bena{\begin{eqnarray}}
\def\eena{\end{eqnarray}}

\def\f(#1/#2){\frac{#1}{#2}}
\def\Frac(#1/#2){\left(\frac{#1}{#2}\right)}
\def\chris(#1-#2-#3){{\mit \Gamma}^{#1}{}_{{#2}{#3}} }
\def\tilchris(#1-#2-#3){\tilde{{\mit \Gamma}}^{#1}{}_{{#2}{#3}}}
\def\hatchris(#1-#2-#3){\hat{{\mit \Gamma}}^{#1}{}_{{#2}{#3}}}

\DeclareFontFamily{OT1}{rsfs}{}
\DeclareFontShape{OT1}{rsfs}{m}{n}{ <-7> rsfs5 <7-10> rsfs7 <10-> rsfs10}{}

\def\hatu{{\hat u}}
\def\hatv{{\hat v}}
\def\hatw{{\hat w}}

\def\hattilu{{\widehat{\tilde u}}}
\def\hattilv{{\widehat{\tilde v}}}
\def\hattilw{{\widehat{\tilde w}}}

\def\HM{{\textrm{HM}}}

\newcommand{\hh}{{ {\hat{h}}}}
\newcommand{\hhone}{\stackrel{(1)}{ {\hat{h}}}}
\newcommand{\hhtwo}{\stackrel{(2)}{ {\hat{h}}}}
\newcommand{\hone}{\stackrel{(1)}{h}}
\newcommand{\htwo}{\stackrel{(2)}{h}}
\newcommand{\mone}{\stackrel{(1)}{m}}
\newcommand{\mtwo}{\stackrel{(2)}{m}}
\newcommand{\hphione}{\stackrel{(1)}{ {{\phi}}}}
\newcommand{\hphitwo}{\stackrel{(2)}{ {{\phi}}}}

\renewcommand{\ptc}[1]{}
\renewcommand{\ptcheck}[1]{}
\renewcommand{\redad}[2]{#2}
\begin{document}

\title{On the energy of asymptotically Horowitz-Myers metrics\footnote{Preprint UWThPh-2019-19}}

\author{Hamed Barzegar}
\email[]{hamed.barzegar@univie.ac.at}
\thanks{supported   by    the Austrian Science Fund (FWF) project No. P29900-N27. }
\affiliation{Faculty of Physics, 
University of Vienna}

\author{Piotr T.\ Chru\'{s}ciel}
\email[]{piotr.chrusciel@univie.ac.at}
\homepage[]{\\ http://homepage.univie.ac.at/piotr.chrusciel}
\thanks{supported in part by  the Austrian Science Fund (FWF) under {project P29517-N27} and by the Polish National Center of Science (NCN) under grant 2016/21/B/ST1/00940.}
\affiliation{Faculty of Physics, 
University of Vienna}

\author{Michael H\"orzinger}
\email[]{mi.hoerz@gmail.com}
\affiliation{Faculty of Physics, 
University of Vienna}

\author{Maciej Maliborski}
\email[]{maciej.maliborski@univie.ac.at}
\affiliation{Faculty of Physics, 
University of Vienna}

\author{Luc Nguyen}
\email[]{nguyenl@maths.ox.ac.uk}
\affiliation{Mathematical Institute and St. Edmund Hall, University of Oxford}

\date{\today}

\begin{abstract}
We present results supporting the Horowitz-Myers conjecture, that the Horowitz-Myers metrics minimise energy amongst metrics on the same spacetime manifold and with the same asymptotic behaviour.
 \ptc{title changed, abstract made more precise}
\end{abstract}


\maketitle


\section{Introduction}

\redad{18VIII19}{
The question of positivity of the total mass in general relativity has played a central role in many investigations in mathematical and theoretical physics, with clear implications for stability, at least at a heuristic level, but also to issues such as the classification of static black hole solutions of Einstein equations~\cite{ChCo,bunting:masood} or  the resolution of the Yamabe problem~\cite{KhuriMarquesSchoen}.

The standard setting in this context is that of asymptotically Minkowskian spacetimes, but other asymptotic conditions are  also of interest. For instance,
Asymptotically Locally Hyperbolic (ALH) metrics have found a prominent place in models of theoretical physics (cf., e.g., \cite{deHaro:2000xn,RyuTakayanagi} and references therein). These models display a fascinating interplay between the positivity of energy and topology: Solutions with spherical conformal infinity are known to have positive mass~\cite{ChDelayHPET,Wang,ChHerzlich};  some classes of solutions with toroidal infinity are known to have positive mass~\cite{Wang} while  some classes of static solutions with toroidal or higher-genus infinity necessarily  have negative mass~\cite{GallowayWoolgar,LeeNeves}. The fact that a toroidal conformal geometry at infinity can be filled-in by a conformally compact filling in many ways~\cite{Anderson98} adds further ambiguities to the subject.
}

The Horowitz-Myers  (HM) solutions of vacuum Einstein equations with a negative cosmological constant,
\redad{18VIII19}{
also known as AdS solitons,
}
provide  an interesting example of  ALH strictly static metrics with total mass which is negative when compared with that of a locally maximally symmetric hyperbolic metric with the same toroidal conformal geometry at infinity.
\redad{18VIII19}{
The $(n+1)$-dimensional spacetimes manifolds underlying the HM metrics are $\R\times\R^2\times\T^{n-2}$, where $\T^{n-2}$ denotes an $(n-2)$-dimensional torus, with a conformal boundary at infinity diffeomorphic to $\R\times S^1\times\T^{n-2}= \R\times  \T^{n-1}$.
The explicit form of the metric can be found in Appendix~\ref{Ap4III18.1}, where we also provide some associated formulae as needed for the remainder of this work.
}

It has been conjectured in \cite[Conjectures 2 and 3]{HorowitzMyers} that, in spacetime dimension five and in
\cite[Section~4]{ConstableMyers} in all dimensions,
the negativity of the relative mass is a normalisation artifact, and that all solutions of the general relativistic constraint equations with the same manifold (``bulk'') topology and same conformal geometry at infinity will have energy larger than the corresponding Horowitz-Myers solution.
 \redad{18VIII19}{
The  relevance of the HM conjecture to  issues in AdS/CFT  was explained clearly in~\cite{HorowitzMyers,ConstableMyers,WoolgarRigidityHM,PageHM}, where it  is related  to consistency of the whole AdS/CFT program, is shown to imply uniqueness   of the AdS soliton~\cite{WoolgarRigidityHM}, and is relevant for  a proper understanding of phase transitions~\cite{PageHM}.
}

The first aim of this paper is to show that,
\redad{18VIII19}{
consistently with the above conjectures,
the (relative) mass of asymptotically HM solutions of the vacuum Einstein equations is indeed positive} for some classes of time-symmetric $n$-dimensional solutions of the general relativistic constraint equations, namely for $U(1)^{n-1}$ symmetric metrics with an orthogonally-transitive $U(1)^{n-2}$ subgroup, and for a subclass of
 $U(1)^{n-2}$ symmetric metrics with an orthogonally-transitive $U(1)^{n-2}$-action.
 
Next, recall that calculations supporting the conjecture  have been  given in~\cite{HorowitzMyers,ConstableMyers}  for
linearised perturbations of HM metrics satisfying a transverse-traceless condition. However, the most general linearised perturbations of HM metrics do not belong to this class, as follows from the non-vanishing of the right-hand side of \eqref{28II18.4} below for general perturbations. The second aim of this paper is to provide  evidence for positivity of energy for general small perturbations, though it should be recognised  that we fail to give a rigorous proof. On the way we   fill some gaps in the arguments of~\cite{HorowitzMyers,ConstableMyers} for the restricted class of perturbations considered in these references, and promote their linearised-perturbations argument to a result about small perturbations.


\section{Preliminaries}

Let $(M, \backg)$ be a smooth $n$-dimensional Riemannian manifold, $n \ge  2$.
\redad{19VII19}{
\emph{Static potentials}
}
are defined as   functions $V$ on $ M$  satisfying
 \begin{equation}\label{12VIII1711-}
    \bcov_i \bcov_ j V = V \left(\bR_{ij} -   \frac {\bR}{n-1}  \backg_{ij} \right)
\,.
\end{equation}
When $\backg$ has constant scalar curvature, an equivalent form is
\begin{equation}\label{12VIII1711}
  \Delta_{\backg} V + \lambda V = 0
  \,,
   \quad
    \bcov_i \bcov_ j V = V (\bR_{ij} - \lambda \backg_{ij} )
\,,
\end{equation}
for some constant $\lambda  \in \R$.
 Here $\bR_{ij}$ denotes the Ricci tensor of the metric
$\backg$, $\bcov$ the Levi-Civita connection of $\backg$, and $\Delta_{\backg} =
 \bcov^k \bcov_k$
is the Laplacian of
$\backg$.

When $\lambda<0$,  rescaling $\backg$ by a constant factor if necessary   we can without loss of generality assume that
$$\lambda=-n
$$
so that
\ptcheck{20XII17}
$$
 \bR := \backg^{ij}\bR_{ij} = \lambda (n-1)=-n(n-1)
 \,,
$$
and this normalisation will be often chosen. This is equivalent to setting $\ell=1$ in the HM metric\ \eqref{24II18.1} below and elsewhere in our equations.

Ignoring an overall dimension-dependent constant, we  use the
\redad{19IX19, and footnote expanded on 2X19}{
widely accepted definition%
\footnote{The formula we use for mass assigns the same number to a metric as the formulae due to Abbott-Deser~\cite{AbbottDeser}, Ashtekar-Magnon~\cite{AshtekarMagnonAdS}, de Haro et al.~\cite{deHaro:2000xn}, or the definition used in~\cite{ConstableMyers}, when simultaneously defined. Indeed, the fact that the Abbott-Deser mass coincides with the Hamiltonian mass of \cite{ChAIHP} follows by inspection of the formulae; the fact that our mass coincides with the Hamiltonian mass has been shown in~\cite[Appendix~B]{CJL}; the fact that the Abbott-Deser mass coincides with the holographic mass is~\cite[Equation (5.22)]{ChruscielSimon} (the formula there has been written in   $3+1$-dimensions but the calculation and the conclusion are  the same in all higher dimensions); the fact that the Ashtekar-Magnon formula coincides with the Hamiltonian one has been shown in~\cite{BCHKK}; the fact that the mass used in~\cite{ConstableMyers} coincides with the Abbott-Deser one follows immediately from the construction of the Abbott-Deser mass.}
of the   mass $m$   of a Riemannian metric $g$ asymptotic to a metric $\backg$ with a static potential $V$ in an equivalent, flexible form given in~\cite{ChHerzlich}:%
}

 \ptc{the referee correctly pointed out that the typesetting looks weird in preprint form, but is actually the right one for the two-column format of PRD here and in some other places}
\begin{multline}\label{2V18.11}
  m = \lim_{R\to\infty} \int_{r=R}
   \big[
   V  g^{ m j}  g^{ i\ell }
    \left(
      \bcov{_ m} g_{ j\ell }
     -
     \bcov{_\ell } g_{ j m}
    \right)
    \\
    +
       (g^{ m j}  g^{ k i}
      - g^{ i  j}  g^{ km}
      )
     (g_{ j m} -\backg_{jm})\bcov{_ k}V
      \big]
      d\sigma_i
      \,.
\end{multline}
\redad{2VII19}{
The limit as $R\to\infty$ of the integrand of the mass,  whenever it exists, will be referred to as the \emph{mass aspect}.
}

\section{Positivity for selected classes  of metrics}\label{s19VI19.1}

\redad{27VIII19}
{Throughout this section, for convenience we rescale  the metric by a constant factor so that $\ell=1$.}
\subsection{Positivity for a class of  $U(1)^{n-1}$ invariant metrics}
 \label{ss19VI19.12}

Consider a metric on
$$
M = \R^2\times \T^{n-2}:=\R^2 \times \underbrace{S^1 \times \cdots \times S^1}_{n-2 \ \mathrm{factors}}
$$
which is invariant under rotations of the $\R^2$ factor as well as rotations of each of the factors $S^1$ of the torus $\T^{n-2}$. In coordinates adapted to the symmetry it can be written in the form
  \begin{equation}
    \label{19VI19.1s}
    g = g_{ij}dx^i dx^j
    \,,
  \end{equation}
where all the $g_{ij}$'s depend only upon
\redad{2VII19}{
the polar radial coordinate on $\R^2$, which will be denoted by $r$.
}
A redefinition of the coordinates
$$
x^a \mapsto x^a +f^a(r)
$$
where $x^1=r$, $x^a = (\theta, x^A)$, with $a,b=2,3,\ldots,n$ and $A,B=3,\ldots,n$,
 with suitably chosen functions $f^a$ allows one to obtain $g_{ra}=0$, bringing the metric to the
form
  \begin{align}
    g
    	&= e^{2u}dr^{2} + e^{2v}d\theta^{2} +
    	2 g_{\theta A} d\theta dx^{A} + g_{AB}dx^{A}dx^{B}
\nonumber    	
\\
    	&\equiv e^{2u}dr^{2} + g_{ab}dx^{a}dx^{b}\,,
   \label{eq:4}
  \end{align}
where $u,v,g_{AB},g_{\theta A}$ are functions of $r$ only. A calculation gives
  \begin{align}
    R &=
    	-e^{-2u}\left( \hat W^{2} + 2\partial_{r}\hat W
      - 2 \hat W\partial_{r}u \right)
\nonumber
\\
     &\quad -  \frac{1}{4}e^{-2u}g^{ab}g^{cd}\partial_{r}g_{ac}\partial_{r}g_{bd}\,,
  \label{eq:6}
  \end{align}
  where
 %
  \begin{equation}
    \label{eq:5}
    \hat W := \frac 12  g^{ab}\partial_{r}g_{ab}
    \,.
  \end{equation}

We will specialise to the case where the orbits of the $U(1)^{n-2}$-isometry subgroup acting on the torus factor of $M$ are orthogonal to the $\R^2$ factor;  this is sometimes referred to as
\emph{orthogonal-transitivity}
and is, at least locally, equivalent to the condition that each of the covector fields, say $X_{(A)}^\flat$,
 $A= 3,\ldots, n$, associated with the Killing vectors  $X_{(A)}$ generating the $U(1)^{n-2}$ action on the torus factor of $M$  satisfies
\begin{equation}\label{19VI19.23}
  dX_{(A)}^\flat \wedge X_{(3)}^\flat\wedge \cdots \wedge X_{(n)}^\flat = 0
  \,.
\end{equation}
\redad{23VI19}{
In other words, the metric is $U(1)^{n-1}$-invariant with an orthogonally-transitive $U(1)^{n-2}$ subgroup.
}
In this case there exist coordinates in which the metric takes the form
\begin{equation}
    \label{19VI19.22}
    g = e^{2u}dr^{2} + e^{2v}d\theta^{2} +
    g_{AB}dx^{A}dx^{B}\,,
  \end{equation}
  where $u$, $v$ and the $g_{AB}$'s  still depend only upon $r$. We then have
   \begin{align}
    R &= -2e^{-2u}\Big(\partial_{r}^{2}v - \partial_{r}u\partial_{r}v
\nonumber
     \\
    & \quad +(\partial_{r}v)^{2} + W \partial_{r}(v-u) +
     \frac 1 2 W^{2} + \partial_{r}W \Big)
\nonumber
     \\
    & \quad- \frac{1}{4} e^{-2u} g^{AB}g^{CD}\partial_{r}g_{AC}\partial_{r}g_{BD}\,,
    \label{19VI19.15}
  \end{align}
  with
  \begin{equation}
    \label{20VI19.1}
    W := \frac{1}{2}g^{AB}\partial_{r}g_{AB}\,.
  \end{equation}

We wish to show that metrics in this class have positive mass
\redad{2VII19}{
with respect to their asymptotic HM background (with $V=r/\ell$ in the coordinates of \eqref{24II18.1} below),
}
whenever the scalar curvature satisfies
\begin{equation}\label{14VI19.5}
  R \ge -n(n-1)
  \,.
\end{equation}
This condition is equivalent to the hypothesis of positivity of energy density for time-symmetric general relativistic initial data sets with negative cosmological constant $\Lambda = -n(n-1)/2$.

Indeed, we claim:
\begin{theorem}
 \label{T19VI19.1}
Consider a metric $g$ on $\R^2 \times \T^{n-2}$ of the form   \eqref{19VI19.22} where all the metric functions depend only upon $r$ and which has a well defined total mass $m$ with respect to a  Horowitz-Myers metric. If the Ricci scalar $R$ of $g$ satisfies \eqref{14VI19.5} then
$$
 m\ge 0
 \,,
$$
vanishing if and only if $g$ coincides with its asymptotic Horowitz-Myers metric.
\end{theorem}

The reader is referred to \cite{ChHerzlich} for the detailed asymptotic conditions needed for a well-defined mass.

\medskip

{\noindent \sc Proof:}
It turns out that the proof is most transparent for   metrics of the form
\begin{equation}
g = e^{2u}dr^2 + e^{2v} d\theta^2 + e^{2w}\big((dx^3)^2+\cdots+(dx^{n})^2
 \big)
  \,,
  	\label{Eq:07VI19-Met}
\end{equation}
where $u$, $v$ and $w$ are functions of $r$. We will therefore first carry out the proof in this case. For the metric \eqref{Eq:07VI19-Met} we have
 \ptcheck{28 III 18 with Luc, and the formula for R rechecked by MH in dim 3 with mathematica, and by MM in dim 3, 4 and 5 with mathematica in VI19}
%
\begin{align}
R
	&= 2e^{-2u}\Big[ - v'' - (n-2)w'' + u'v' + (n-2)u'w'
	 \nonumber\\
		&\qquad\qquad - (n-2) v'w' - (v')^2 - \frac{(n-1)(n-2)}{2} (w')^2\Big]
	\;.\label{Eq:07VI19-Curv}
\end{align}
%

%
%

Suppose that $g$ asymptotes to a Horowitz-Myers metric \eqref{24II18.1} with parameter $r_0$.
We will denote  this asymptotic background as $g_{\HM,r_0}$,
\redad{18VIII19}{
the associated function $u$ as in \eqref{Eq:07VI19-Met} by $u_{\HM,r_0}$, etc.} We have
\begin{align*}
u_{\HM,r_0}
	&= -\ln r - \frac{1}{2}\ln(1 - \frac{r_0^n}{r^n})
	\;,\\
v_{\HM,r_0}
	&= \ln r + \frac{1}{2}\ln(1 - \frac{r_0^n}{r^n})
	\;,\\
w_{\HM,r_0}
	&= \ln r
	\;.
\end{align*}
It is readily seen that $R(g_{\HM,r_0}) = -n(n-1)$.
In order to obtain a smooth metric  at $r = r_0$, $\theta$ needs to have period ${4\pi}/{nr_0}$.
 \ptcheck{29 III}

Using e.g.\ the perturbation arguments in~\cite{ChDelayAH,CGNP}, in order to prove positivity it suffices to assume that the components of the metric, when expressed in terms of orthonormal  (ON) frame of the  asymptotic background, behave as  $ r^{-n}$ plus $o(r^{-n})$ terms, and that this behaviour is preserved under differentiation.
We can therefore, without loss of generality, assume the asymptotic expansions, for large $r$,
\begin{align}
 u
	&= u_{\HM,r_0} + \hatu = u_{\HM,r_0} + u_n \,r^{-n} + o(r^{-n})
	\;,
\nonumber
\\
 v
	&= v_{\HM,r_0} + \hatv = v_{\HM,r_0} + v_n\,r^{-n}  + o(r^{-n})
	\;,
\nonumber
\\
 w
	&= w_{\HM ,r_0} + \hatw = w_{\HM ,r_0} + w_n\, r^{-n} + o(r^{-n})
	\;,
\label{18VI19.1}
\end{align}
 where $u_n$, $v_n$ and $w_n$ are constants. 
 In order to determine  the mass aspect, rather than calculating the integrand of \eqref{2V18.11} one can proceed as follows:  We calculate
\begin{widetext}
   \ptcheck{29 III, power corrected}
%
\begin{align}
R
	&= 2e^{-2u_{\HM,r_0}}(1 - 2\hatu + O(r^{-2n})) \times
 \nonumber
\\
		&\qquad \times \Big[-\frac{n(n-1)}{2} e^{2u_{\HM,r_0}}  - (\hatv'' + (n-2)\hatw'')
            + [u_{\HM,r_0}'(\hatv' + (n-2)\hatw')
						+ \hatu'(v_{\HM,r_0}' + (n-2)w_{\HM,r_0}')]
 \nonumber
\\
			&\qquad\qquad - (n-2) (v_{\HM,r_0}' \hatw'+ \hatv'\,w_{\HM,r_0}')
			- 2v_{\HM,r_0}' \hatv' - (n-1)(n-2) w_{\HM,r_0}' \hatw' + O(r^{-2(n+1)})\Big]
 \nonumber
\\
	&= - n(n-1) + 2n(n-1)\hatu + 2r^2\Big[  - (\hatv'' + (n-2)\hatw'')
                +  \frac{n-1}{r} \hatu'
				-   \frac{n+1}{r} \hatv'
				- \frac{(n+1)(n-2)}{r}\hatw'\Big] + O(r^{-2n})
 \nonumber
\\
	&= - n(n-1) + 2r^{\red{1-n}}\frac{d}{dr}\Big[ (n-1) r^{n}\hatu -  r^{n+1} \hatv' - (n-2)r^{n+1}\hatw'\Big] + O(r^{-2n})
		\;.\label{5VI19.2}
\end{align}
\end{widetext}
\redad{3VII19}{
It follows from this last equation and from the way that the mass is calculated in \cite{ChHerzlich}
that the mass aspect function of $g$ relative to $g_{\HM }$  is the leading term of the expression inside the square bracket.
}
Hence, up to a positive multiplicative constant,
%
\begin{equation}\label{5VI19.21}
\Theta(g) = \frac{2(n-1)}{n}u_n + 2v_n + 2(n-2)w_n
	\;.
\end{equation}
The next step of our analysis consists of redefining the coordinate $r$ to a new coordinate 
$\tilde r$ so that the function $\hat u$ in the new coordinate system
vanishes.
In other words, the function $u$ in the new coordinate system will be equal  $u_{\HM,\tilde r_0}$ for some $\tilde r_0$: 

\begin{lem}\label{Lem:07VI19-L1}
There exists a smooth increasing function $r \mapsto \tilde r(r)$ such that
\begin{align}
\tilde r(r) &= r  -\frac{r_0^n-\tilde r_0^n + 2 u_n}{2 n r^{n-1}} +
o(r^{1-n}) \text{ as }r \rightarrow \infty
	\;,\label{Eq:07VI19-C1}
	\\
g_{\tilde r \tilde r}
	&= e^{2 u_{\HM ,\tilde r_0}} \text{ with } \tilde r_0 = \tilde r(r_0) > 0
	\;.\label{Eq:07VI19-C2}
\end{align}
\end{lem}

\begin{proof}
Define
\[
F(r) = \int_1^r e^{u_{\HM ,1}(\xi)}\,d\xi
	\;.
\]
Then there exists a constant $F_n$ depending upon the space dimension $n$ such that
\[
F(r) = \ln r + F_n
- \frac{1 }{2nr^n}
 +o(r^{-n}) \text{ as } r \rightarrow \infty
	\;.
\]
Let $\tilde r_0 > 0$ be such that
\[
\lim_{r \rightarrow \infty} \Big[\int_{r_0}^r e^{u(\xi)}\,d\xi - \ln r\Big] = F_n - \ln \tilde r_0
	\;.
\]
The desired function $\tilde r$ is then defined by
\begin{equation}
 F\Big(\frac{\tilde r(r)}{\tilde r_0}\Big) = \int_{r_0}^r e^{u(\xi)}\,d\xi
 	\;.\label{Eq:07VI19-E1}
\end{equation}

We proceed to check \eqref{Eq:07VI19-C1}-\eqref{Eq:07VI19-C2}. Indeed, we have by the definition of $\tilde r_0$ that
\[
\int_{r_0}^r e^{u(\xi)}\,d\xi =  \ln r + (F_n - \ln \tilde r_0)
- \frac{r_0^n + 2 u_n}{2nr^n}
 +o(r^{-n}) \text{ as } r \rightarrow \infty
	\;.
\]
Using \eqref{Eq:07VI19-E1} and

\[
F\Big(\frac{\tilde r(r)}{\tilde r_0}\Big) =  \ln \tilde r + (F_n - \ln \tilde r_0) - \frac{\tilde r^n_0 }{2n\tilde r^n}
 +o(\tilde r^{-n}) \text{ as } \tilde r \rightarrow \infty
\]
yields \eqref{Eq:07VI19-C1}. On the other hand, from \eqref{Eq:07VI19-E1}, we have
\[
\int_{r_0}^r e^{u(\xi)}\,d\xi
	=  \int_1^{\frac{\tilde r(r)}{\tilde r_0}} e^{u_{\HM ,1}(\xi)}\,d\xi
	=\int_{\tilde r_0}^{\tilde r(r)} e^{u_{\HM ,\tilde r_0}(\xi)}\,d\xi
		\;.
\]
Differentiating in $r$ yields
\[
e^{u(r)}\,dr = e^{u_{\HM ,\tilde r_0}(\tilde r)}\,d\tilde r,
\]
which gives \eqref{Eq:07VI19-C2}.
\qedskip
\end{proof}

Using the variable $\tilde r$ given in Lemma \ref{Lem:07VI19-L1}, we rewrite the metric \eqref{Eq:07VI19-Met} as
\begin{equation}
g = e^{2\tilde u}d \tilde{r}^2 + e^{2\tilde v} d\theta^2 + e^{2\tilde w}\big(( {dx^3)^2+\cdots+(dx^{n}})^2
 \big)
  \,,
  	\label{Eq:07VI19-Mettil}
\end{equation}
where $\tilde u$, $\tilde v$ and $\tilde w$ are functions of $\tilde r \in [\tilde r_0, \infty)$, keeping in mind that $\theta$ is an angular variable with period $\frac{4\pi}{nr_0}$.

Write
\begin{align*}
\tilde v
	&= v_{\HM,\tilde r_0} + \hattilv = v_{\HM,\tilde r_0} + \tilde v_n\,\tilde r^{-n}  + o(\tilde r^{-n})
	\;,\\
\tilde w
	&= w_{\HM ,\tilde r_0} + \hattilw = w_{\HM ,\tilde r_0} + \tilde w_n\, \tilde r^{-n} + o(\tilde r^{-n})
	\;,
\end{align*}
 where $\tilde v_n$ and $\tilde w_n$ are constants,  and note that   the function $\hattilu:=\tilde u - u_{\HM,\tilde r_0} $ is identically zero by construction. It is readily seen that
   \ptcheck{8VI19}
 %
 \begin{align*}
\tilde v_n
	&= v_n +  \frac{1}{2}(\tilde r_0^n - r_0^n) {+ \frac{r_0^n-\tilde r_0^n + 2 u_n}{2 n }}
	\;,\\
\tilde w_n
	&=  w_n  { + \frac{r_0^n-\tilde r_0^n + 2 u_n}{2 n }}
	\;.
\end{align*}
Hence, by \eqref{5VI19.21},
\ptcheck{8VI19}
%
\begin{equation}\label{Eq:07VI19-M1}
\Theta(g) = -\frac{1}{n} (\tilde r_0^n - r_0^n)  + 2\tilde v_n + 2(n-2)\tilde w_n
	\;.
\end{equation}
This implies
\begin{align}
\Theta(g)
	&= -\frac{1}{n} (\tilde r_0^n - r_0^n)
		+ \lim_{\tilde r \rightarrow \infty} 2\tilde r^n(\hattilv + (n-2)\hattilw)
		\nonumber
		\\
	&= -\frac{1}{n} (\tilde r_0^n - r_0^n)
		-  \lim_{\tilde r \rightarrow \infty} \frac{2}{n}\tilde r^{n+1} (\hattilv + (n-2)\hattilw)'
		\nonumber
		\\
	&= -\frac{1}{n} (\tilde r_0^n - r_0^n)
		+ \frac{2}{n} \int_{\tilde r_0}^\infty  \Big[(-  {\hattilv}
  - (n-2)  {\hattilw})' \Wpsi {\tilde r^{n+1}}\Big]'\,d\tilde r
	\label{Eq:07VI19-M2}
	\;,
\end{align}
where $'$ now stands for $\frac{d}{d\tilde r}$ and $\Wpsi$ is any function of $\tilde r$ which vanishes at $\tilde r = \tilde r_0$ and $\Wpsi \rightarrow 1$ as $\tilde r \rightarrow \infty$.

Recall formula \eqref{Eq:07VI19-Curv}, which in the current coordinate system translates to
\begin{widetext}
\begin{align}
R
	&= 2e^{-2\tilde u}\Big[ - \tilde v'' - (n-2)\tilde w'' + \tilde u'\tilde v' + (n-2)\tilde u'\tilde w'
 - (n-2) \tilde v' \tilde w' - (\tilde v')^2 - \frac{(n-1)(n-2)}{2} (\tilde w')^2\Big]
	\;.\label{Eq:07VI19-Curvtil}
\end{align}
Using the fact that $g_{\HM ,\tilde r_0}$ has curvature $-n(n-1)$, we thus have
\ptcheck{9VI19 up to here}
%
\begin{align}
 \nonumber
R + n(n-1)
	&= 2e^{-2u_{\HM,\tilde r_0}}\Big[ - \hattilv'' - (n-2)\hattilw''
			+ (u_{\HM,\tilde r_0}' - 2 v_{\HM,\tilde r_0}' - (n-2)w_{\HM,\tilde r_0}')\hattilv'
		\\
		&\quad  + (n-2)(u_{\HM,\tilde r_0}' -  v_{\HM,\tilde r_0}' - (n-1) w_{\HM,\tilde r_0}')\hattilw'
			- (n-2)\hattilv' \hattilw' - (\hattilv')^2
			- \frac{(n-1)(n-2)}{2} (\hattilw')^2\Big]
	\;.
\label{14VI19.21}
\end{align}
Using $\Wpsi = \tilde r^{-(n+1)}e^{2v_{\HM,\tilde r_0} + (n-1)w_{\HM,\tilde r_0} + \hattilv}$ in \eqref{Eq:07VI19-M2} and noting that $u_{\HM,\tilde r_0} = -v_{\HM,\tilde r_0}$, we arrive at
\ptcheck{10VI19}
%
\begin{align}
 \nonumber
\Theta(g)
	&= -\frac{1}{n} (\tilde r_0^n - r_0^n)
		+ \frac{2}{n} \int_{\tilde r_0}^\infty  \Big[ (- \hattilv - (n-2)\hattilw)' e^{2v_{\HM,\tilde r_0} + (n-1)w_{\HM,\tilde r_0} + \hattilv }\Big]'\,d\tilde r
\\
 \nonumber
	&= 	-\frac{1}{n} (\tilde r_0^n - r_0^n)
		+  \frac{2}{n} \int_{\tilde r_0}^\infty e^{2v_{\HM,\tilde r_0} + (n-1)w_{\HM,\tilde r_0} + \hattilv}\Big[- \hattilv'' - (n-2)\hattilw''
\\
 \nonumber
		&\qquad - (\hattilv' + (n-2)\hattilw')(2v_{\HM,\tilde r_0}' + (n-1)w_{\HM,\tilde r_0}' + \hattilv')\Big]\,d\tilde r
\\
 \nonumber
	&=	-\frac{1}{n} (\tilde r_0^n - r_0^n)
		+ \frac{2}{n} \int_{\tilde r_0}^\infty e^{2v_{\HM,\tilde r_0} + (n-1)w_{\HM, \tilde r_0} + \hattilv}\Big\{\frac{1}{2}e^{2u_{\HM,\tilde r_0}}(R + n(n-1))
	\\
		&\qquad + (v_{\HM,\tilde r_0}' - w_{\HM,\tilde r_0}')\hattilv'
			+ \frac{(n-1)(n-2)}{2} (\hattilw')^2
		 \Big\}\,d\tilde r
	\;.
\label{19VI19.18}
\end{align}
\end{widetext}
The term containing $\hattilv'$ can be computed as follows:
\ptcheck{10VI19}
%
\begin{align}
& \frac{2}{n} \int_{\tilde r_0}^\infty e^{2v_{\HM,\tilde r_0} + (n-1)w_{\HM,\tilde r_0} + \hattilv}(v_{\HM,\tilde r_0}' - w_{\HM,\tilde r_0}')\hattilv'\,d\tilde r
\nonumber
\\
	&= \tilde r_0^n \int_{\tilde r_0}^\infty  e^{\hattilv}\hattilv'\,d\tilde r = \tilde r_0^n e^{\hattilv}\Big|_{\tilde r = \tilde r_0}^{\tilde r = \infty} = \tilde r_0^n (1 - e^{\hattilv(\tilde r_0)})
	\;.
 \label{18VI19.11}
\end{align}
As $g$ is regular at $\tilde r = \tilde r_0$ and $\theta$ has period $\frac{4\pi}{nr_0}$, we have that
\ptcheck{10VI19, one could consider adding some details here}
%
\begin{equation}\label{20VI19.2}
e^{\hattilv(\tilde r_0)} = \frac{r_0}{\tilde r_0}
 \,.
\end{equation}
Altogether we obtain
%
\begin{align}
 \nonumber
\Theta(g)
	&=	-\frac{1}{n} (\tilde r_0^n - r_0^n) + \tilde r_0^{n-1}(\tilde r_0 - r_0)
	\nonumber
\\
		&\qquad
		+ \frac{2}{n} \int_{\tilde r_0}^\infty e^{2v_{\HM,\tilde r_0} + (n-1)w_{\HM, \tilde r_0} + \hattilv}
	\nonumber
\\
		&\qquad\qquad				
		\Big\{\frac{1}{2}e^{2u_{\HM,\tilde r_0}}(R + n(n-1))
	\nonumber
\\
		&\qquad\qquad
	     \quad + \frac{(n-1)(n-2)}{2} (\hattilw')^2
		 \Big\}\,d\tilde r
	\;.\label{Eq:07VI19-M3}
\end{align}

The quantity $-\frac{1}{n} (\tilde r_0^n - r_0^n) + \tilde r_0^{n-1}(\tilde r_0 - r_0)$ is non-negative due to the convexity of the function $t \mapsto t^n$, which establishes that $\Theta(g)$ is positive or vanishes.

The case $m=0$ implies $\Theta(g) =0$,  and \eqref{Eq:07VI19-M3} gives
$$
   \hattilw' \equiv 0 \equiv R + n(n-1)
   \,,
   \quad
   r_0 = \tilde r_0
   \,.
$$
We see that $ \hattilw\equiv 0 =\hattilw_n=\hattilv_n$ as well, and \eqref{14VI19.21}   gives
\begin{align} 0
	&= - \hattilv''
			+ (u_{\HM,\tilde r_0}' - 2 v_{\HM,\tilde r_0}' - (n-2)w_{\HM,\tilde r_0}'- \hattilv')\hattilv'
	\;,
\label{22VI19.4}
\end{align}
while from \eqref{20VI19.2} we obtain
\begin{equation}\label{20VI19.3}
e^{\hattilv(\tilde r_0)} = 1
 \,.
\end{equation}
The maximum principle shows that $\hattilv\equiv 0$, and
we have proved:

\begin{prop}
 \label{P14VI19.1}
If the metric $g$ in \eqref{Eq:07VI19-Met} satisfies $R \geq -n(n-1)$ then $g$ has non-negative mass, vanishing if and only if $g$ coincides with the corresponding Horowitz-Myers metric.
\qed
\end{prop}

We now pass to general $U(1)^{n-1}$-orthogonally-transitive-invariant metrics \eqref{19VI19.22}.  For this let us write
\begin{equation}\label{19VI19.14}
  \partial_r g_{AB} = \frac{2 W}{n-2} g_{AB} + H_{AB}
  \,,
\end{equation}
with $W$ as in \eqref{20VI19.1},
  thus $g^{AB} H_{AB}=0$. This allows us to rewrite the last term appearing in the formula \eqref{19VI19.15} for the Ricci scalar of $g$ as
%
\begin{widetext}
  \begin{equation}
    - \frac{1}{4} e^{-2u} g^{AB}g^{CD}\partial_{r}g_{AC}\partial_{r}g_{BD}
     =
      - \frac{1}{4} e^{-2u} \left( g^{AB}g^{CD}H_{AC}H_{BD} + \frac{4 W^2}{n-2}
     \right)
     \le -  \frac{  e^{-2u} W^2}{n-2}
      \,.
      \label{19VI19.26}
  \end{equation}
Inserting this into    \eqref{19VI19.15} one obtains
  \begin{equation}
    R
    \le   -2e^{-2u}\left(
    \partial_{r}^{2}v - \partial_{r}u\partial_{r}v+(\partial_{r}v)^{2}
     + W \partial_{r}(v-u) +
   \frac 1 2 W^{2} + \partial_{r}W
  \right)
   -  \frac{  e^{-2u} W^2}{n-2}
    \,.
      \label{19VI19.16}
  \end{equation}
\end{widetext}
Defining
\begin{equation}\label{19VI19.13}
  w':= \frac{W}{ n-2 }
  \,,
\end{equation}
the inequality~\eqref{19VI19.16} can be rewritten as
\begin{eqnarray}
 R \le  2 e^{-2u}\Big[ - v'' - (n-2)w'' + u'v' + (n-2)u'w'
	 \nonumber\\
		 - (n-2) v'w' - (v')^2    -\frac{(n-2)(n-1)}{2}(w')^2
 \Big]
 \,.
 \nonumber
 \\
  \label{7VI19.1+}
\end{eqnarray}
This coincides with   \eq{Eq:07VI19-Curvtil} except that the equality there is changed to an inequality consistent with what we need to prove. With the definition \eqref{19VI19.13} the formula \eqref{5VI19.21} (derived as  the mass aspect of the metric \eqref{Eq:07VI19-Met}) provides also the correct formula for the metric \eqref{19VI19.22}. The argument of the proof  of Proposition \ref{P14VI19.1} leads again to \eqref{19VI19.18} and \eqref{Eq:07VI19-M3} with the equalities there replaced by $\ge$, which establishes that $m\ge 0$.

If $m=0$ all the inequalities arising in the argument have to be equalities, in particular \eqref{19VI19.26} with $\le$ replaced by an equality implies that $\partial_r g_{AB}$ is pure trace, and  Proposition~\ref{P14VI19.1} applies. The proof of Theorem~\ref{T19VI19.1} is complete.
\qed 

\subsection{Positivity for a class of orthogonally-transitive-$U(1)^{n-2}$-invariant metrics}
 \label{ss19VI19.11}

It turns out that the arguments given so far partially generalise to metrics which are invariant under an orthogonally-transitive action of $U(1)^{n-2}$ by isometries of the torus factor of $M$. Such metrics can be written in the form
\begin{equation}
  \label{21VI19.1}
  g = e^{2u}dr^{2} + e^{2v}d\theta^{2} +
  g_{AB}dx^{A}dx^{B}\,,
\end{equation}
where $u,v,g_{AB}$ are functions of $(r,\theta)$. One finds
\begin{widetext}
\begin{multline}
  R = -2e^{-2u}\left(\partial_{r}^{2}v - \partial_{r}u\partial_{r}v+(\partial_{r}v)^{2} + \partial_{r}(v-u)W^{r} +
    \frac{1}{2}(W^{r})^{2} + \partial_{r}W^{r} \right) 
  - \frac{1}{4} e^{-2u} g^{AB}g^{CD}\partial_{r}g_{AC}\partial_{r}g_{BD}
  \\
  -2e^{-2v}\left(\partial_{\theta}^{2}u - \partial_{\theta}u\partial_{\theta}v+(\partial_{\theta}u)^{2} -\partial_{\theta}(v-u)W^{\theta} +
    \frac{1}{2}(W^{\theta})^{2} + \partial_{\theta}W^{\theta} \right) 
  - \frac{1}{4} e^{-2v} g^{AB}g^{CD}\partial_{\theta}g_{AC}\partial_{\theta}g_{BD}\,,
   \label{20VI19.9}
\end{multline}
\end{widetext}
with
\begin{equation}
  \label{eq:3}
  W^{r} = \frac{1}{2}g^{AB}\partial_{r}g_{AB},
  \quad
  W^{\theta} = \frac{1}{2}g^{AB}\partial_{\theta}g_{AB}\,.
\end{equation}

A useful device in the $U(1)^{n-1}$-symmetric case was the introduction of a new radial coordinate $\tilde r$ so that $g_{\tilde r \tilde r}$ takes a canonical form. This does not seem to go through in the general case above while preserving a form of the metric which is convenient for the remaining arguments.
On the other hand, the proof generalises if we assume at the outset that
\begin{equation}\label{20VI19.11}
 u \equiv u_{\HM}
 \,.
\end{equation}
We then have:

\begin{theorem}
 \label{T20VI19.1}
Consider a metric $g$ on $\R^2 \times \T^{n-2}$ of the form   \eqref{21VI19.1} where the metric functions depend only upon $(r,\theta)$ and which has a well defined total mass $m$ with respect to a  Horowitz-Myers background metric. If \eqref{20VI19.11} holds and if the Ricci scalar $R$ of $g$ satisfies $R \ge -n(n-1)$ then
$$
 m\ge 0
 \,,
$$
vanishing if and only if $g$ coincides with its asymptotic Horowitz-Myers metric.
\end{theorem}

\proof
As before, the proof is most transparent for metrics of the form
\begin{equation}
g = e^{2u}dr^2 + e^{2v} d\theta^2 + e^{2w}\big(({dx^3)^2+\cdots+(dx^{n}})^2
 \big)
  \,,
  	\label{20VI19.12}
\end{equation}
where we allow now $u$, $v$ and $w$ to depend both upon $r$ and $\theta$. Then
\ptcheck{18VI19; calculated by luc by hand and checked with MM's mathematica calculation in a few dimensions}
\begin{widetext}
\begin{align}
 \nonumber
R
	&= 2e^{-2u}\Big[ - v_{,rr} - (n-2)w_{,rr} + u_{,r} v_{,r} + (n-2)u_{,r}w_{,r}
		 - (n-2) v_{,r}w_{,r} - (v_{,r})^2 - \frac{(n-1)(n-2)}{2} (w_{,r})^2\Big]
	\\
		&\qquad + 2e^{-2v}\Big[ - u_{,\theta\theta} - (n-2) w_{,\theta\theta}
			+ u_{,\theta} v_{,\theta}  + (n-2) v_{,\theta} w_{,\theta}
			- (n-2) u_{,\theta} w_{,\theta}
			- (u_{,\theta})^2 - \frac{(n-1)(n-2)}{2} (w_{,\theta})^2
			\Big]
	\;.
 \label{20VI19.8}
\end{align}
If we assume \eqref{20VI19.11} and write
\[
v  = v_{\HM} + \hatv \text{ and } w = w_{\HM} + \hatw
\]
then, using \eqref{14VI19.21},
 \ptcheck{checked, crossref added and first lines commented out}
\begin{align*}
R + n(n-1)
	&= 2e^{-2u_\HM}\Big[ - \hatv_{,rr} - (n-2)\hatw_{,rr}
		+ (u_\HM' - 2 v_\HM' - (n-2)w_\HM')\hatv_{,r}
		\\
		&\qquad\qquad + (n-2)(u_\HM' -  v_\HM' - (n-1) w_\HM')\hatw_{,r}
		 \\
		&\qquad\qquad - (n-2)\hatv_{,r} \hatw_{,r} - (\hatv_{,r})^2
			- \frac{(n-1)(n-2)}{2} (\hatw_{,r})^2\Big]
			\\
		&\qquad \qquad + 2(n-2)e^{-2v_\HM - 2\hatv}\Big[  -  \hatw_{,\theta\theta}
			+  \hatv_{,\theta} \hatw_{,\theta}
			- \frac{n-1}{2} (\hat w_{,\theta})^2
			\Big]
	\;.
\end{align*}
\end{widetext}
Already-mentioned perturbation arguments allow us to assume that $g$ is asymptotic to $g_{\HM }$ in the sense of \eqref{18VI19.1}, where the expansion coefficients are now allowed to depend upon $\theta$. We have
 \ptcheck{18VI19}
\begin{widetext}
\begin{align*}
\Theta(g)
	&=\lim_{r \rightarrow \infty} 2r^n(\hatv + (n-2)\hatw)
     = -  \lim_{r \rightarrow \infty} \frac{2}{n}r^{n+1} \partial_r (\hatv + (n-2)\hatw)
	= \frac{2}{n} \int_{r_0}^\infty  \partial_r [\partial_r (- \hatv - (n-2)\hat w) \Wpsi r^{n+1}]\,dr
\,,
\end{align*}
\end{widetext}
where
 $\Wpsi$ is any function which vanishes at $r = r_0$ and $\Wpsi \rightarrow 1$ as $r \rightarrow \infty$. In the sequel we take
 $\Wpsi = {r^{-(n+1)}}e^{2v_{\HM} + (n-1)w_{\HM} + \hatv}$. Using $u_{\HM} = -v_{\HM}$ we find
 \ptcheck{18VI19}
%
\begin{align*}
\Theta(g)
	&= \frac{2}{n} \int_{r_0}^\infty  \partial_r [\partial_r (- \hatv - (n-2)\hat w) e^{2v_{\HM} + (n-1)w_{\HM} + \hatv }]\,dr\\
	&= \frac{2}{n} \int_{r_0}^\infty e^{2v_{\HM} + (n-1)w_{\HM} +  {\hat v}}\Big[- \hatv_{,rr} - (n-2)\hat w_{,rr}
	\\
		&\qquad - (\hatv_{,r} + (n-2)\hat w_{,r}) \times
	\\	
	   & \qquad \times (2v_{\HM}' + (n-1)w_{\HM}' + \hatv_{,r})\Big]\,dr\\
	&= \frac{2}{n} \int_{r_0}^\infty e^{2v_{\HM} + (n-1)w_{\HM} + \hatv}\Big\{\frac{1}{2}e^{2u_{\HM}}(R + n(n-1))
	\\
		&\qquad
{+}  ( v_{\HM}' - w_{\HM}')\hatv_{,r}
			+ \frac{(n-1)(n-2)}{2} (\hatw_{,r})^2
		\\
		&\qquad - (n-2)e^{-4v_\HM - 2\hatv} \times
\\
		&\qquad				
		\times \Big[  -  \hatw_{,\theta\theta}
			+  \hatv_{,\theta} \hatw_{,\theta}
			- \frac{n-1}{2} (\hat w_{,\theta})^2 \Big]\Big\}\,dr
 \,.
\end{align*}
Regularity of the metric at the core geodesic $r = r_0$ requires $\hatv = 0$ there.
As in \eqref{18VI19.11} we have now
 \ptcheck{18VI19}
\begin{align*}
& \frac{2}{n} \int_{r_0}^\infty e^{2v_{\HM} + (n-1)w_{\HM} + \hat v}(v_{\HM}' - w_{\HM}')\hatv_{,r}\,dr\\
	&= r_0^n \int_{r_0}^\infty  e^{\hatv}\hatv_{,r}\,dr = r_0^n e^{\hatv}\Big|_{r = r_0}^{r = \infty} = 0
	\;.
\end{align*}
The term containing $-  \hatw_{,\theta\theta} +  \hatv_{,\theta} \hatw_{,\theta}$ cancels away after integrating over $\theta$:
\begin{equation}\label{20VI19.21}
\int e^{-\hatv}(-\hatw_{,\theta\theta} + \hatv_{,\theta}\hatw_{,\theta})\,d\theta = \int \partial_\theta (-e^{-\hat v} \hatw_{,\theta})\,d\theta = 0
 \,.
\end{equation}
So, under the hypotheses of Theorem~\ref{T20VI19.1}, we obtain non-negativity of mass for metrics satisfying moreover \eqref{20VI19.12}. The vanishing of the mass implies $\hat w \equiv 0$ and one concludes by an argument similar to that in the proof of Proposition~\ref{P14VI19.1}.

We consider now the general case. Formula \eqref{20VI19.9} with $\hat u \equiv 0$ implies
\begin{widetext}
\begin{align}
  R &\le -2e^{-2u}\left(\partial_{r}^{2}v - \partial_{r}u\partial_{r}v+(\partial_{r}v)^{2} + \partial_{r}(v-u)W^{r}
    +
    \frac{(n-1)}{2(n-2)}(W^{r})^{2}
     + \partial_{r}W^{r} \right)
\nn  
  \\
  & \quad -2e^{-2v}\left(\partial_{\theta}^{2}u - \partial_{\theta}u\partial_{\theta}v+(\partial_{\theta}u)^{2} -\partial_{\theta}(v-u)W^{\theta} +
    \frac{(n-1)}{2(n-2)}(W^{\theta})^{2} + \partial_{\theta}W^{\theta} \right)
 \,,
   \label{20VI19.4}
\end{align}
where we have used  \eqref{19VI19.26}. Introducing
\begin{equation}\label{20VI19.5}
  w_{,r}= \frac{ W^r}{n-2}
\end{equation}
allows us to rewrite \eqref{20VI19.4} as
\begin{align}
 \nonumber
  R  	&\le 2e^{-2u}\Big[ - v_{,rr} - (n-2)w_{,rr} + u_{,r} v_{,r} + (n-2)u_{,r}w_{,r}
		- (n-2) v_{,r}w_{,r} - (v_{,r})^2 - \frac{(n-1)(n-2)}{2} (w_{,r})^2\Big]
  \\
  &
  \quad -2e^{-2v}\left(\partial_{\theta}^{2}u - \partial_{\theta}u\partial_{\theta}v+(\partial_{\theta}u)^{2} -\partial_{\theta}(v-u)W^{\theta} +
    \frac{(n-1)}{2(n-2)}(W^{\theta})^{2} + \partial_{\theta}W^{\theta} \right)
 \,.
   \label{20VI19.4+}
\end{align}
\end{widetext}
This differs from \eqref{20VI19.8} by the replacement of an equality sign by $\le$, and replacement of $w_{,\theta}$ by $W^\theta/(n-1)$. The rest of the proof requires only trivial changes, for example \eqref{20VI19.21} is replaced by
\begin{equation}\label{20VI19.22}
\int e^{-\hatv}(-W_{,\theta } + \hatv_{,\theta}W^{\theta})\,d\theta = \int \partial_\theta (-e^{-\hat v} W^{\theta})\,d\theta = 0
 \,.
\end{equation}
The details are left to the reader.
\qed
%


\section{Perturbations of static Riemannian metrics}

In the remainder of this work we wish to address the question of positivity of the relative mass for small perturbations of the Horowitz-Myers metrics. We start with some general considerations.

Given a metric $g$ asymptotic to a background metric $\backg$, we define
\begin{eqnarray}
&
 h_{ij}:= g_{ij} - \backg_{ij}
 \,,
 &
 \label{17VI19.6}
\\
 &
 \psi^j:=\bcov{_ i}  g^{ij}
 \quad
 \Longleftrightarrow
 \quad
 g^{ij} \bcov_i h_{j\ell} =-g_{\ell j}  \psi^j
 \,,
  \label{28II18.1}
  &
\\
 &
 \phi := g^{ij} h_{ij}
 \quad
 \Longrightarrow
 \quad
  \overline \phi := \backg^{ij} h_{ij} = \phi +  O\left(|h |^2_{\backg} \right)
 \,.
  &
  \notag
  \\
    \label{28II18.1+}
\end{eqnarray}
We will denote by $\check h$, respectively by $\hat h$, the $g$-trace-free, respectively the $\backg{}$-trace-free, part of $h$:
%
\begin{equation}\label{20II18.1}
 \check{h}
 :=
 h - \frac{1}{n}  {\phi} \, g
 \,,
 \qquad
 \hat{h}
 :=
 h - \frac{1}{n} \overline{\phi} \, \backg
  \,.
\end{equation}
The most relevant fields for our purposes here are $\phi$ and $\hat h$, and we emphasise  that $\check h$ and $\hat h$ differ by terms quadratic in $\varepsilon$ if   $  h = O(\varepsilon)$   and if $\varepsilon$ is small; similarly for $\phi$ and $\bar \phi$.
We also use the notation
\begin{equation}
 g^{ij}
 =
 \backg^{ij} - h^{ij} + \chi^{ij}
 \,,
 \label{2V18.1}
\end{equation}
where
\begin{widetext}
\begin{equation}
  h^{i j} = \backg^{i k} \backg^{j \ell } h_{k \ell}
  \quad
  \text{and}
  \quad
 \chi^{ij} :=
    \backg^{ik}\backg^{j\ell}\backg^{mn}h_{km}h_{n\ell}
       +O(|h |^3 _{\backg})
        =
       O(|h|_{\backg}^2)
 \,.
  \label{2V18.1+}
\end{equation}
\end{widetext}
In order to address the question of gauge-freedom, it would be convenient to apply a diffeomorphism to $g$ so that
\begin{equation}\label{13II18.1a}
 \bcheckpsi^i:= \psi^i + \frac{1}{2} g^{i k} \bcov_k \phi
\end{equation}
vanishes.  The equation $\bcheckpsi^i = 0$ will be referred to as the \emph{harmonicity condition}, which is motivated by the fact that it reduces to the harmonic-coordinates condition in the case of a flat background. It is likely that the harmonicity condition can be achieved in whole generality for perturbations of a Horowitz-Myers background, but this is irrelevant for the current work as our analysis is inconclusive anyway.

In~\cite{BCNarxiv} the following formula was derived,
which holds for any asymptotically hyperbolic background $(M,\backg)$ with a static potential $V$, under the usual conditions for existence of the hyperbolic mass $m$:
\ptcheck{26II18 and rechecked the Riemann term 28 II 18, also  agrees with a direct calculation for second order perturbations in the HMCompetitor.tex paper and in mathematical file HM/Christoffels.nb, June 2019}
\begin{eqnarray}
 m
 &
  \displaystyle
  =
 \int_M \Big[ & (R-\overline R) V
  + \Big(
 \quadratic{
   \frac{n+2}{8n}
   |\bcov \, \phi|^2_{\backg}
  +
   \frac{1}{4}
   |\bcov \hat{h}|^2_{\backg}}
   \nn
   \\
   &&
   \quadratic{-
     \frac{1}{2}
    \hat h^{i \ell} \hat h^{j m}  \adsR{_{\ell m i j}}
    -
    \frac{n+2}{2n} \phi \hat h^{i j} \adsR{_{ i j}}
     -
         \frac{n^2-4}{8n^2}  \lambda
      \phi^2
      }
\nn
  \\
    &&
    \gauge{
   -
    \frac{1}{2}
    \big(
       |\bcheckpsi|^2_{\backg}
       -
       \bcheckpsi^i\bcov_i \phi )
   }
 \Big)
 V
 +
  \Big(
    \tensor{h}{^k_i} \bcheckpsi^i
    +
    \frac 12
    \phi\bcheckpsi^k
   \Big)
   \bcov_k V
  \nn
\\
 & &
     +
    \Big(\herr
    +
    \hdhsqerr\Big)V
    \nn
 \\
 &&
 +
 \hdherr
  | \bcov  V|_{\backg}
 \Big]
  \sqrtbg
  \,.
   \label{12III18.1}
\end{eqnarray}
\redad{30VIII19}{
In this equation all indices are raised and lowered
}
 using the background metric $\backg$.

We will also need another formula from
\cite{BCNarxiv}:
\ptcheck{28II18}
\begin{equation}\label{12VIII17.3+b}
 R
 =
 \adsR{_{ij}} g^{ i j}
  +
   \bcov{_ k} \big[ g^{ i j}  g^{ k\ell }
    \left(
      \bcov{_ i} h_{ j\ell }
    -
    \bcov{_\ell } h_{ j i}
    \right)
    \big]
    +
    Q
    \,,
\end{equation}
where
 \ptcheck{15 VIII  the result with luc up to here (but not all intermediate formulae) and rechecked by Hamed on 1 II 18 and rechecked including all intermediate formulae}
 \begin{eqnarray}
 Q &:= &
 \frac 14   g^{ i j} g^{ k p} g^{\ell q}
  \Big(
  2
    \bcov{_ p} h_{ j\ell }  \bcov{_ q} h_{ k i}
 \nonumber
 \\
   && \quad  -
     \bcov{_\ell } h_{ k p}   \bcov{_ q} h_{ i j}
     -
     \bcov{_ i} h_{ p q}   \bcov{_ j} h_{ k\ell }
  \Big)
  \,.
  \phantom{xxx}
    \label{17XII17.1}
 \end{eqnarray}
%

In the notation of \eqref{28II18.1}-\eqref{13II18.1a}, the identity~\eqref{12VIII17.3+b} becomes
\begin{align}
- \frac 12 \bcov_k \left( g^{kl}D_l \phi
 \right) & =  R -
 \adsR{ }
  +
 \adsR{_{ij} } h^{ij}
\nonumber
\\
 &  \quad -
  \underbrace{
   \bcov{_ k} \left(   g^{ k\ell } h_{ j i}
    \bcov{_\ell } g^{ij} - \bcheckpsi ^k
    \right)
     \underbrace{-Q}_{O(|\bcov h|^2)}
    }_{\textrm{gauge and higher order terms}}
    \,.
\label{28II18.5}
\end{align}
It follows 
that a metric perturbation will satisfy  the linearised time-symmetric scalar constraint equation if and only if
\begin{equation}\label{28II18.4}
- \frac 12 \bcov_k \left( g^{kl}D_l \phi
 \right) -
 \frac {\adsR{  } } n   \phi
   =
 \adsR{_{ij} } \hat h^{ij}%
    \,.
\end{equation}
If $\backg$ is a space-form the  term linear in $\hat h_{ij}$ at the right-hand side vanishes, which typically implies that $\phi$ itself is higher order (compare the discussion before Equation~(2.14) in \cite{ChDelayKlingerNonDegenerate}).
However, this is not true for general $\backg$ and $\hat h$, in particular one cannot assume that $\phi=0$ for general perturbations of the Horowitz-Myers metrics in harmonic gauge.


Let us consider a one-parameter family of perturbations $h$ of the metric of the form
\begin{equation}\label{14VI19.1}
  h_{ij} = \epsilon \hone_{ij} + \epsilon^2 \htwo_{ij} + O(\epsilon^3)
  \,,
\end{equation}
where $\hone$ and $\htwo$ are independent of $\epsilon$. We   assume that this expansion is preserved by differentiation. Subsequently, the mass will have an expansion
\begin{equation}\label{14VI19.2}
 m = \epsilon \mone +  \epsilon^2 \mtwo  + O(\epsilon^3)
  \,,
\end{equation}
We will use similar notation for  expansions of $\hat h$ and $\phi$:
\begin{equation}\label{14VI19.1h}
  \hh_{ij} = \epsilon \hhone_{ij} + \epsilon^2 \hhtwo_{ij} + O(\epsilon^3)
  \,,
  \quad
   \phi = \epsilon \hphione  + \epsilon^2 \hphitwo + O(\epsilon^3)
   \,.
\end{equation}

Suppose, first, that $h$ satisfies the constraint equation up to terms of order $\epsilon^2$; equivalently, that $\hone$ satisfies the linearised constraint equation. Dividing \eqref{12III18.1} by $\epsilon$ and passing to the limit $\epsilon=0$ one obtains the well-known result, that static metrics are local extrema of mass on the set of solutions of the constraint equations:
\begin{equation}\label{14VI19.3}
  \mone =0
  \,.
\end{equation}

Suppose, next, that $h$ satisfies the constraint equation up to terms of order $\epsilon^3$ and that the gauge condition $\bcheckpsi^k=0$ holds. Dividing \eqref{12III18.1} now by $\epsilon^2$ and passing to the limit $\epsilon=0$ one obtains
 \ptcheck{27V19, together with MH and HB, and note that some of the terms assume the normalisation of the Ricci scalar}
\begin{eqnarray}
\mtwo
 &
  \displaystyle
  =
 \int_M
  \Big[ &   \Big(
 \quadratic{
   \frac{n+2}{8n}
   |\bcov \, \hphione|^2_{\backg}
  +
   \frac{1}{4}
   |\bcov \hhone|^2_{\backg}}
   \nn
   \\
   &&
    -  \frac{1}{2}
     \hhone{}^{i \ell}  \hhone{}^{j m}  \adsR{_{\ell m i j}}
     - \frac{n+2}{2n} \hphione  \hhone{}^{i j} \adsR{_{ i j}}
   \nn
   \\
   &&
     +
         {\frac{n^2-4}{8n \ell^2}}
      (\hphione)^2
 \Big]
       {V}
  \sqrtbg
  \,.
\label{29V19.1}
\end{eqnarray}
We note that the knowledge of the perturbed metric to first order in $\epsilon$ suffices to obtain a formula for the mass which is accurate to second order in $\epsilon$.

To simplify notation, we will from now on interchangeably write $(\hphione,\hhone)$
and $(\phi,\hat h)$, the smallness parameter $\epsilon$ being implicitly understood whenever required.

\section{Perturbations of Horowitz-Myers metrics}
 \label{s21VI19.1}

If $\backg$ is the space-part of the Horowitz-Myers metric, the curvature-dependent terms in \eqref{29V19.1} read (see Appendix~\ref{Ap4III18.1})
 \ptcheck{6III18, rechecked by HB Nov 2018}
\begin{eqnarray}
  \adsR{_{ij}}\hat h^{ij}
   & = &
   \frac{n}{2 \ell^2} \left( \frac{r_0}{r} \right)^n
   \left(
    \hat h^{\hat 1 \hat 1}
         +
     \hat h^{\hat 2 \hat 2}
   \right)
         \,,
         \label{4III18.1s}
\\
  \adsR{_{ijk\ell}}\hat h^{ik}\hat h^{j\ell}
   &=&
   f''
     \left[ (  \hat h^{\hat 1 \hat 2})^2
     +\frac{1}{4}
     \left(
        \hat h^{\hat 1 \hat 1}
         -
        \hat h^{\hat 2 \hat 2}
     \right)^2
     \right]
     \nonumber
\\
 &&
     +
     \left[
       \frac{(3-n) f}{(n-2)r^2}
       +
       \frac{f'}{r}
       -
       \frac{f''}{4}
     \right]
      \left(
        \hat h^{\hat 1 \hat 1}
         +
        \hat h^{\hat 2 \hat 2}
     \right)^2
\nonumber
\\
     && +
     \frac{f}{r^2} |\,\widehat{\hat{h}} \,|^2_{\backg}
    \,,\label{4III18.2a}
\end{eqnarray}
where
\begin{equation}
 f(r)
 :=
 \frac{r^2}{\ell^2}
 \left[
   1-\left( \frac{r_0}{r}\right)^n
 \right]
 \,,
  \label{27II18.11a}
\end{equation}
\begin{equation}\label{6III18.21a}
  \widehat{\hat h}_{\hat A \hat B}
   :=  {\hat h}_{\hat A \hat B}
  - \frac 1 {n-2} \backg^{\hat C \hat D}
  {\hat h}_{\hat C \hat D}
  \backg_{\hat A \hat B}
   \,,
\end{equation}
with $\widehat{\hat h}_{\hat a \hat i} \equiv 0\equiv \widehat{\hat h}_{\hat a \hat b}$,
\redad{3VII19}{
 where from now on  $a,b\in \{1,2\}$,
}
  and where hatted indices denote frame components with  respect to the $\backg$-orthonormal frame \eqref{27II18.1} below.

The question arises, whether the quadratic form \eqref{29V19.1} in the fields $(\hphione, \hhone_{ij})$ is positive definite. An affirmative answer would establish the Horowitz-Myers conjecture for small perturbations of Horowitz-Myers metrics. We show in Appendix~\ref{ss18VI19.1} that this is not the case: there exist fields so that the right-hand side of \eqref{29V19.1} is negative.

However, the examples there satisfy neither the linearised constraint equations nor the harmonicity conditions, as would have been needed to invalidate the conjecture. And we have neither been able to find fields satisfying all necessary requirements, nor to prove that no such fields with negative $\mtwo$  exist. The analysis that we present below suggests strongly, but fails to prove, that if such fields existed, then there would also exist negative mass configurations depending only upon $r$. Since we have just proved that any metric in the relevant class depending only upon $r$ has positive mass, we are led to suspect that the Horowitz-Myers conjecture remains correct for all small perturbations of the Horowitz-Myers metrics.

An obvious approach to analyse the sign of the right-hand side of \eqref{29V19.1}  is to estimate  $\mtwo$ from below by discarding all positive terms which do not contain radial derivatives of the fields. Thus, from the term $ \frac 14 |\bcov \hat{h}|^2_{\backg}$ there we will only keep the following,
 using \eqref{24XI18.4} below,
\begin{widetext}
\begin{eqnarray}
\nonumber 
 \frac{1}{4} |\bcov \hat{h}|^2_{\backg}
  &\ge &
   \frac{f}{4}
   \big(
   |\partial_r  \hat{h}_{\hat 1 \hat{1}}  |^2
  +
  {2} |\partial_r \hat{h}_{\hat 1 \hat{2}}|^2
  +
   |\partial_r  \hat{h}_{\hat 2 \hat{2}} |^2
    {
    +
      |\, \partial_r\widehat{\hat{h}} \,|^2_{\backg}
     }
   \big)
\\
    &=&
   \frac{f}{4}
   \big(
   \frac{1}{2} |\partial_r (\hat{h}_{\hat 1 \hat{1}} + \hat{h}_{\hat 2 \hat{2}})|^2
  +
   {2} |\partial_r \hat{h}_{\hat 1 \hat{2}}|^2
  +
   \frac{1}{2} |\partial_r (\hat{h}_{\hat 1 \hat{1}}- \hat{h}_{\hat 2 \hat{2}})|^2
    {
    +
      |\, \partial_r\widehat{\hat{h}} \,|^2_{\backg}
     }
    \big)
     \,.
\end{eqnarray}
This
   gives, after ignoring further irrelevant-looking obviously-positive terms,
    \ptcheck{ 7VII19, $\ell=1$, so I changed $\lambda=-n$, checked for consistency\\ -- \\ should be rechecked, has been a while...}
%
\begin{eqnarray}
 \mtwo
 &
  \displaystyle
  \ge
 \int_M \Big[ &
 \Big(
 \quadratic{
   \frac{n+2}{8n}
   f |\partial_r \phi|^2
  +
   \frac{f}{4}
   \big(
   \frac{1}{2} |\partial_r (\hat{h}_{\hat 1 \hat{1}} + \hat{h}_{\hat 2 \hat{2}})|^2}
  \quadratic{
  +
   {2} |\partial_r \hat{h}_{\hat 1 \hat{2}}|^2
  +
   \frac{1}{2} |\partial_r (\hat{h}_{\hat 1 \hat{1}}- \hat{h}_{\hat 2 \hat{2}})|^2
     {
    +
      |\, \partial_r\widehat{\hat{h}} \,|^2_{\backg}
     }
   \big)
 }
    \nonumber
\\
%
 &&
 \quadratic{
   -
     \frac{f''}{2}
     \left[ (  \hat{h}_{\hat 1 \hat 2})^2
     +\frac{1}{4}
     \left(
        \hat{h}_{\hat 1 \hat 1}
         -
        \hat{h}_{\hat 2 \hat 2}
     \right)^2
     \right]
     \nonumber
     }
  \quadratic{
     -\frac 12
     \left[
       \frac{(3-n) f}{(n-2)r^2}
       +
       \frac{f'}{r}
       -
       \frac{f''}{4}
     \right]
      \left(
        \hat{h}_{\hat 1 \hat 1}
         +
        \hat{h}_{\hat 2 \hat 2}
     \right)^2
      }
\nn
  \\
   &&
   \quadratic{
 -
   \frac{n+2}{4 \ell^2}\left( \frac{r_0}{r} \right)^n
   \left(
    \hat{h}_{\hat 1 \hat 1}
         +
     \hat{h}_{\hat 2 \hat 2}
   \right) \phi
    +
         {\frac{n^2-4}{8n \ell^2}}
      \phi^2
   }
   { -
     \frac{f}{2r^2} |\,\widehat{\hat{h}} \,|^2_{\backg}
     }
\Big]
 {V}
  \sqrtbg
  \,,
   \label{19XI18.1}
\end{eqnarray}
with equality attained on those perturbations with $h_{rA}=0=h_{\theta A}$ which depend only upon $r$.
\end{widetext}

\subsection{Positive contribution from  $h_{\hat 1 \hat 2}$}
Let us, first, consider those terms in \eqref{19XI18.1} which involve $h_{\hat 1 \hat 2}$.
Using
\begin{align}
  f'' &= \ell^{-2}\Big(
  2 - (n-2)(n-1) \frac{r^n_0}{r^n}
   \Big)
  \,,
  \quad
  \sqrt{\det \backg} =  \frac{r^{n-2}}{\ell^{n-2}}
  \,,
  \notag
  \\
   V &= \frac r \ell
   \,,
 \label{20XI18.1}
\end{align}
%
we need to analyse the integral
\begin{equation}\label{20XI18.2}
   \int_{r_0}^\infty
   \big(
    \frac f{2} |\partial_r h_{\hat 1 \hat 2}|^2
    -
    \frac {f''}2 |
     h_{\hat 1 \hat 2}|^2
     \big)
     \frac{r^{n-1}}{\ell^{n-1}}
     \, dr
     \,.
\end{equation}
To this end, we use the identity
\ptcheck{20XI18}
\begin{equation}\label{20XI18.3}
   \int_{r_0}^\infty
  2 f \partial_r h_{\hat 1 \hat 2}  h_{\hat 1 \hat 2}
    r^{n-2}
     \, dr
   + \int_{r_0}^\infty
  |h_{\hat 1 \hat 2} |^2
   (r\partial_r f + (n-2) f) r^{n-3}
     \, dr
     = 0
     \,.
\end{equation}
Hence, for any $\alpha \in \mathbb{R}$,
\begin{widetext}
\begin{eqnarray}
 \lefteqn{
\int_{r_0}^\infty
   \big(
    f |\partial_r h_{\hat 1 \hat 2}|^2
    -
    f'' |
     h_{\hat 1 \hat 2}|^2
     \big)
     r^{n-1}
     \, dr
    =
    \int_{r_0}^\infty
   \big(
    f |\partial_r h_{\hat 1 \hat 2}|^2
    -
    f'' |
     h_{\hat 1 \hat 2}|^2
     \big)
     r^{n-1}
     \, dr
      }
      &&
      \nonumber
     \\
     &&
     + \alpha\int_{r_0}^\infty
  2 f \partial_r h_{\hat 1 \hat 2}  h_{\hat 1 \hat 2}
    r^{n-2}
     \, dr
   + \alpha \int_{r_0}^\infty
  |h_{\hat 1 \hat 2} |^2
   (r\partial_r f + (n-2) f) r^{n-3}
     \, dr
     \nonumber
     \\
     &=&
    \int_{r_0}^\infty
   \Big[
    f \big|\partial_r h_{\hat 1 \hat 2} + \frac{\alpha}{r} h_{\hat 1 \hat 2}\big|^2\
    +
    W_\alpha(f) |
     h_{\hat 1 \hat 2}|^2
     \Big]
     r^{n-1}
     \, dr
      \,,
       \label{24XI18.21}
\end{eqnarray}
\end{widetext}
where
 \ptcheck{20XI18, the function Walpha and its lower bound rechecked 22VI19}
\begin{align*}
W_\alpha(f)
	&=  -f'' + \frac{\alpha}{r} \partial_r f + \frac{(n-2)\alpha - \alpha^2}{r^2}  f \\
	&= \frac{1}{\ell^2}\Big( -(\alpha^2 - n\alpha + 2) + (n^2 - 3n + 2 + \alpha^2)\big(\frac{r_0}{r}\big)^n\Big).
\end{align*}
Choosing $\alpha = \frac{n}{2}$
 we have $W_{n/2}(f) \geq  (n^2-8)/(4\ell^2) >0$ for $n\ge 3$, and so with this choice the integral in \eqref{20XI18.2} is non-negative.

The reader will note that  the argument leading to \eq{24XI18.21} establishes the following inequality, for any function $\zeta$ and any $\alpha\in \R$,
\begin{equation}\label{24XI18.22}
  \int_M f
  (\partial_r\zeta)^2 \sqrtbg
   \ge
  \int_M
  \big(
   \alpha ( r\partial_r f  +(n-2) f) - \alpha^2 f
    \big)
      r^{-2}\zeta^2  \sqrtbg
   \,.
\end{equation}

\begin{rem}
 \label{R21VI19}
 {\rm
An identical calculation applies to those terms in \eqref{19XI18.1} which involve $\hat h_{\hat{1}\hat{1}}
-  \hat h_{\hat{2}\hat{2}}$. However,  these terms are coupled with the remaining ones through   the harmonicity condition, and their positivity compensates for the negativity of the remaining contributions in the radial case, so one should not discard them when estimating $\mtwo$ from below.
}
\end{rem}

\subsection{Positive contribution from  $\widehat{\hat h}_{\hat A \hat B}$}
\label{ss21VI19.11}

We consider now those terms in \eqref{19XI18.1} which explicitly involve $\widehat{\hat h}_{\hat A \hat B}$. These are, up to irrelevant numerical factors and an inessential integration over the remaining variables,
\begin{equation}
  \int_{r_0}^ \infty
  f \big(
      |\, \partial_r\widehat{\hat{h}} \,|^2_{\backg}
  -
     \frac{2}{ r^2} |\,\widehat{\hat{h}} \,|^2_{\backg}
  \big)
  \frac{r^{n-1}}{\ell^{n-1}} dr
  \,.
   \label{19XI18.100}
\end{equation}
Using the inequality
$$
 \big| \,
     \partial_r  |\,  \widehat{\hat{h}} \,| _{\backg}
  \big|
  \le
   |\, \partial_r\widehat{\hat{h}} \,| _{\backg}
   \,,
$$
the analysis of the sign of \eqref{19XI18.100} can be reduced to that of the sign of the integral
\begin{equation}
  \int_{r_0}^ \infty
  f \big(
     (\partial_r \zeta)^2
  -
     \frac{2}{ r^2} \zeta^2
  \big)
  \frac{r^{n-1}}{\ell^{n-1}} dr
  \,,
   \label{19XI18.101}
\end{equation}
for differentiable functions $\zeta$.
A calculation as in \eqref{24XI18.21} with $f''$ there replaced by $2f/r^2$ gives, for any $\alpha \in \mathbb{R}$,
\begin{eqnarray}
 \lefteqn{
\int_{r_0}^\infty
   \big(
    f |\partial_r h_{\hat 1 \hat 2}|^2
    -
   \frac 2{ r^2} |
     h_{\hat 1 \hat 2}|^2
     \big)
     r^{n-1}
     \, dr
      }
      &&
      \nonumber
     \\
     &=&
    \int_{r_0}^\infty
   \Big[
    f \big|\partial_r h_{\hat 1 \hat 2} + \frac{\alpha}{r} h_{\hat 1 \hat 2}\big|^2\
    +
    \tilde W_\alpha(f) |
     h_{\hat 1 \hat 2}|^2
     \Big]
     r^{n-1}
     \, dr
      \,,
\notag
\\      
       \label{24XI18.21a}
\end{eqnarray}
where
\ptcheck{28VIII by mm}
\begin{align}
 \nonumber
\tilde W_\alpha(f)
	&=  \frac{\alpha}{r} \partial_r f + \frac{(n-2)\alpha - \alpha^2-2}{r^2}  f \\
	&= \frac{1}{\ell^2}\Big( -(\alpha^2 - n\alpha + 2) + (\alpha^2  +2)\big(\frac{r_0}{r}\big)^n\Big).
\end{align}
The choice $\alpha = \frac{n}{2}$
leads similarly to  $\tilde W_{n/2}(f) \geq  (n^2-8)/(4\ell^2) >0$ for $n\ge 3$, which shows that the integral in \eqref{19XI18.100} is non-negative. 

\subsection{The remainder}
 \label{ss21VI19.500}

After discarding those fields which have been shown to give a positive contribution to the mass so far, and ignoring the warning in Remark~\ref{R21VI19}, one is left to face  a lower bound for $\mtwo$ governed by the integral
\begin{eqnarray}
\lefteqn{
I :=
 \int_{r_0}^\infty
  \Big[
 \quadratic{
   \frac{n+2}{8n}
   f |\partial_r \phi|^2
  +
   \frac{f}{8}  |\partial_r (\hat{h}_{\hat 1 \hat{1}} + \hat{h}_{\hat 2 \hat{2}})|^2
 }
 }
 &&
  \nn
\\
&&
  \quadratic{
     -\frac 12
     \left[
       \frac{(3-n) f}{(n-2)r^2}
       +
       \frac{f'}{r}
       -
       \frac{f''}{4}
     \right]
      \left(
        \hat{h}_{\hat 1 \hat 1}
         +
        \hat{h}_{\hat 2 \hat 2}
     \right)^2
      }
\nn
  \\
   &&
   \quadratic{
 -
   \frac{n+2}{4 \ell^2}\left( \frac{r_0}{r} \right)^n
   \left(
    \hat{h}_{\hat 1 \hat 1}
         +
     \hat{h}_{\hat 2 \hat 2}
   \right) \phi
    +
           {\frac{n^2-4}{8n \ell^2}}
      \phi^2
   }
 \Big]
     \frac{r^{n-1}}{\ell^{n-1}}
     \, dr
  \,.
  \nonumber
  \\
   \label{20XI18.10}
\end{eqnarray}
As already mentioned, we show in Appendix~\ref{ss18VI19.1} that $I$ can take negative values, when viewed as a functional of
$$
 \xi:= h_{\hat 1 \hat 1} + h _{\hat 2 \hat 2}
$$
after taking $\phi=0$, in all dimensions $n\ge 3$.
However, after enforcing the linearised scalar constraint equation,
 \ptcheck{27V19 with HB and MH}
\begin{equation}\label{23XI18.1a1}
- \frac 12 \bcov_{\hat k} \left( \backg^{\hat k \hat l}\bcov_{\hat l} \phi
 \right) + \frac{(n-1)}{\ell^2}  \phi
     =
   \frac{n}{2 \ell^2} \left( \frac{r_0}{r} \right)^n (\hat h_{\hat 1 \hat 1} + \hat h _{\hat 2 \hat 2})
         \,,
\end{equation}
the fields $\phi$ and $\xi$ are not independent anymore. For example, if $\phi\equiv 0$ we obtain $\xi\equiv0$, then $I= 0$, and from what has been said so far $\mtwo\ge 0$ for such variations. We recover a result already observed by Horowitz and Myers \cite{HorowitzMyers}, that $\mtwo$ is positive for transverse-traceless perturbations of the metric satisfying the linearised constraint equations.

On the other hand, if $n=3$ and if we take
$\phi(r)=r^{-5}$, Equation~\eqref{23XI18.1a1} gives
 \ptcheck{3VII19 and corrected by MM}
\begin{equation}
  \label{eq:1}
  \xi(r) = \frac{30-11r^3}{3r^{5}}
  \,,
\end{equation}
%
%
and $I
=-79/84$.

When $n=4$ a negative value of $I\approx -0.575$
is obtained by setting $\phi(r)= \tanh(r)/r^6$, which results in
\begin{widetext}
\begin{equation}
  \xi(r) = \frac{3 \left(7-2 r^4\right) \tanh (r)+r \left(4 r^4+\left(r^4-1\right) r \tanh
      (r)-6\right) \text{sech}^2(r)}{2 r^6}
  \,.
\end{equation}
\end{widetext}
We note that the ansatz $\phi(r) =r^{-c }$ leads to a positive $I$ in  all  dimensions that we tried, namely $4\le n \le 16$, regardless of the choice of the exponent $c $ for which the integral converges.

\subsection{A functional on $\phi$}

The above does not invalidate the HM conjecture in space-dimensions three and four because imposing   harmonicity leads to a non-vanishing field $h_{\hat 1 \hat 1} - h_{\hat 2\hat 2}$, the contribution of which restores positivity of $\mtwo$ for $r$-dependent perturbations. Indeed, if we prescribe a function $\phi(r)$, we can then calculate $\xi$ from \eqref{23XI18.1a1}. If we consider metric variations in harmonic gauge satisfying
\begin{equation}\label{24XI18.a5b}
  \hat{h}_{\hat 1 \hat 2}=0
  =
  \hat{h}_{\hat 1 \hat A}
  =
  \hat{h}_{\hat 2 \hat A}
  \,,
\end{equation}
with all metric perturbations depending only upon $r$, then the harmonicity conditions reduce   to the equation \eqref{25II19.3} for $\hat h _{\hat 1 \hat 1}$,
\ptcheck{29V; corrected with HB 22V19}
\begin{equation}\label{24XI118.a5}
  \sqrt f \partial_r
   \big(\hat h_{\hat 1 \hat 1} - \frac {n-2}{2n}  \phi\big)
   =  \frac{f'}{2 \sqrt f}
 (\hat h_{\hat 2 \hat 2} - \hat h_{\hat 1 \hat 1})
 -
  {
    \frac{\sqrt f}{r }
    (
    \xi + (n-2)
 \hat h_{\hat 1 \hat 1} )
 }
  \,,
\end{equation}
where $\hat h_{\hat 2 \hat 2}$ is viewed as a function of the already-known field $\xi$ and of $\hat h_{\hat 1 \hat 1}$:
\begin{equation}\label{7V19.1}
  \hat h_{\hat 2 \hat 2}= \xi-\hat h_{\hat 1 \hat 1}
  \,.
\end{equation}
In space-dimension $n=3$ we then necessarily have
$$
 \hat h_{\hat 3 \hat 3}= -\xi
 \,,
$$
and \eqref{24XI118.a5} becomes
\ptcheck{29V both PTC and HB and MH}
\begin{equation}\label{24XI118.a5a}
    r^{-(n-2)} \partial_r
 \big( r^{n-2}f \hat h_{\hat 1 \hat 1}\big)
   =
   \left(\frac{f'}{2 }
  {-} \frac{  f}{r } \right) \xi
 +   \frac {n-2}{2n} f  {\partial_r}\phi
  \,,
\end{equation}
Integrating from $r_0$, regularity at $r_0$ enforces the solution to be
\begin{widetext}
\begin{equation}\label{24XI118.a5c}
    \hat h_{\hat 1 \hat 1}
   = r^{-(n-2)}  f^{-1} \int_{r_0}^r
   \left(
   \big(\frac{f'}{2 }
 {-} \frac{  f}{r } \big) \xi
 +   \frac {n-2}{2n} f  {\partial_r }\phi
 \right)\Big|_{r=s} 
 s^{n-2} \, ds
  \,.
\end{equation}
\end{widetext}
%
The function $\hat h_{\hat 2 \hat 2}$ can now be determined using \eqref{7V19.1}, and we obtain a linearised metric perturbation satisfying the gauge conditions and the linearised constraint equation.

The  right-hand side of \eqref{19XI18.1}  with $\widehat{\hat h}_{AB}= h_{r\theta}=0$ becomes thus a functional of $\hphione$,
 the positivity of which follows in an indirect way from Proposition~\ref{P14VI19.1}. However, positivity when $\hphione$ is allowed to depend upon all variables is not clear.

 \appendix

\section{Horowitz-Myers metrics}
 \label{Ap4III18.1}
 
The Horowitz-Myers ``soliton'' metric  $\tilde{g}_{\textrm{\HM }}$
  reads
\begin{align}
 \tilde{g}_{\textrm{\HM }}
 & =
 -
 \frac{r^2}{\ell^2} dt^2
 +
 \frac{\ell^2}{r^2}
 \frac{dr^2}{1-\left( \frac{r_0}{r}\right)^n}
\nonumber
\\
& \quad +
 \frac{r^2}{\ell^2}
 \left[
  1-\left( \frac{r_0}{r}\right)^n
 \right]
 d\theta^2
  + \frac{r^2}{\ell^2}
 \sum_{A=1}^{n-2} (dx^A)^2
 \,,
 \label{24II18.1}
\end{align}
where
\begin{equation}
 \ell^2
 =
 -
 \frac{n(n-1)}{2\Lambda}
 \,,
\end{equation}
and where $r_0>0$
is a constant, $\theta$ is an angle with the period
$$
 \beta = \frac{4\pi \ell^2}{nr_0}
 \,,
$$
and the $x^A$'s,
$A \in \{3,\ldots,n\}$,
are local coordinates on an $(n-2)$-dimensional flat manifold.
Setting
\begin{equation}
 f(r)
 :=
 \frac{r^2}{\ell^2}
 \left[
   1-\left( \frac{r_0}{r}\right)^n
 \right]
 \,,
  \label{27II18.11}
\end{equation}
the spatial part, say $ \backg_{\textrm{\HM }}$, of \eqref{24II18.1} takes the form
\begin{equation}
 \backg_{\textrm{\HM }}
 =
 \frac{dr^2}{f(r)}
 +
 f(r)
 d\theta^2
 +
 \ell^{-2}  r^2 \overline{\delta}
 \,,
\end{equation}
where $\overline{\delta} = \delta_{AB} dx^A dx^B$ is flat.
We will write $g_{\HM,r_0}$ for $g_{\HM}$ when $r_0$ needs to be made explicit.

Dividing the metric by $\ell^2$ and rescaling the coordinates $(t,\theta,x^A)$ suitably we can always achieve $\ell=1$. A subsequent rescaling of $r$ leads to $r_0=1$. This can be used to reduce the analysis to one where
\begin{equation}\label{26XI18.1}
  \ell=1=r_0
  \,,
  \quad
  \beta = \frac{4\pi}{n}
   \,.
\end{equation}
%

The nonvanishing, up to index symmetries, Christoffel symbols read
 \ptcheck{not used at this stage, checked with mathematica in dim from 3 to 10 }
\begin{widetext}
\begin{equation}
 \Bgamma{^r_{rr}}
 =
 - \frac{f'}{2f}
 \,,
 \quad
 \Bgamma{^r_{\theta \theta}}
 =
 - \frac{1}{2} f f'
 \,,
 \quad
  \Bgamma{^r_{A B}}
  =
  - \ell^{-2} r f \delta_{AB}
  \,,
  \quad
    \Bgamma{^\theta_{r \theta}}
    =
    \frac{f'}{2f}
    \,,
    \quad
      \Bgamma{^A_{r B}}
      =
      \frac{1}{r} \delta_{AB}
      \,.
\end{equation}
The nontrivial components, again up to index symmetries, of the Riemann tensor of the (Riemannian) metric $
 \backg_{\textrm{\HM }}$
are (cf., e.g., \cite[Appendix~D]{ACD2})
\begin{equation}
 \adsR{_{r \theta r \theta}}
 =
 -\frac{1}{2} f''
 \,,
 \quad
 \adsR{_{rArB}}
 =
 - \frac{r f'}{2f \ell^2} \delta_{AB}
 \,,
 \quad
 \adsR{_{\theta A \theta B}}
 =
 - \frac{r f' f}{2\ell^2} \delta_{AB}
 \,,
 \quad
 \adsR{_{A B C D}}
 =
 -
 \ell^{-4} r^2 f
  \left(
    \delta_{A C} \delta_{B D}
    -
    \delta_{A D} \delta_{B C}
  \right)
  \,,
\end{equation}
It is convenient to introduce the orthonormal co-frame
\begin{equation}\label{27II18.1}
 {\bar\theta} ^{\hat 1} = \frac{dr}{\sqrt{f}}
  \,,
  \quad
 {\bar\theta} ^{\hat 2} = {\sqrt{f}} d\theta
  \,,
  \quad
 {\bar\theta} ^{\hat A} =  \ell^{-1} r dx^ A
  \,.
\end{equation}

We have the following non-vanishing connection one-forms
 \ptcheck{27V19, in dimension 4 with mathematica, and together with HB by hand 19VI19}
\begin{equation}\label{19XI18.1-}
  -{\bar \omega}_{\hat 1\hat 2} = {\bar \omega}_{\hat 2\hat 1} = \frac{f'}{2\sqrt f}{\bar\theta} ^{\hat 2}
  \,,
   \qquad
  -{\bar \omega}_{\hat 1\hat A} = {\bar \omega}_{\hat A\hat 1} = \frac{\sqrt f}{r}{\bar\theta} ^{\hat A}
   \,,
\end{equation}
with $\hat a,\hat b \in \{\hat1,\hat 2\}$, $\hat A,\hat B \in \{\hat 3,\ldots,\hat n\}$,
 \ptcheck{agrees with Horowitz Myers which give this also in dimension 5}
and
%
\begin{equation}
 \adsR{_{\hat 1\hat 2\hat 1\hat 2}}
 =
 -\frac{1}{2} f''
 \,,
 \quad
 \adsR{_{\hat a \hat A \hat b \hat B}}
 =
 - \frac{  f'}{2r } \backg_{\hat a \hat b}\backg_{\hat A \hat B}
 \,,
 \quad
 \adsR{_{\hat A\hat  B \hat C \hat D}}
 =
 -
  \frac{f }{r^2}
  \left(
    \backg_{\hat A \hat C} \backg_{\hat B \hat D}
    -
    \backg_{\hat A \hat D} \backg_{\hat B \hat C}
  \right)
  \,.
\end{equation}

The nontrivial frame components of the Ricci tensor read
\begin{equation}
 \adsR{_{\hat a \hat b}} =
   - \frac{1}{2}
   \left[
    f''
    +
    (n-2)  {\frac{ f'}{r}}
  \right]   \backg_{\hat a \hat b}
         \,,
         \qquad
  \adsR{_{\hat A\hat B}} =
       -
            \left[
              \frac{f'}{r}
              +
              (n-3) {\frac{ f }{r^2}}
            \right]
            \backg_{\hat A\hat B}
  \,.
\end{equation}

Let us define
\begin{equation}\label{6III18.21}
  \widehat{\hat h}_{\hat A \hat B}
   :=  {\hat h}_{\hat A \hat B}
  - \frac 1 {n-2} \backg^{\hat C \hat D}
  {\hat h}_{\hat C \hat D}
  \backg_{\hat A \hat B}
   \,,
\end{equation}
with $\widehat{\hat h}_{\hat a \hat i} :=0$.
Using this notation, we have
 \ptcheck{6III18, rechecked by HB Nov 2018}
\begin{eqnarray}
  \adsR{_{ij}}\hat h^{ij}
   & = &  - \frac{1}{2}
   \left[
    f''
    +
    (n-4)  {\frac{ f'}{r}}
    -
              2(n-3) {\frac{ f }{r^2}}
  \right]
   \backg_{\hat a \hat b}\hat h^{\hat a \hat b}
 \nonumber
\\
   & =
  &
   \frac{n}{2 \ell^2} \left( \frac{r_0}{r} \right)^n
   \left(
    \hat h^{\hat 1 \hat 1}
         +
     \hat h^{\hat 2 \hat 2}
   \right)
         \,,
         \label{4III18.1}
\\
  \adsR{_{ijk\ell}}\hat h^{ik}\hat h^{j\ell}
   & = &
     - f''\big(
     \hat h^{\hat 1 \hat 1}  h^{\hat 2 \hat 2} -  (\hat h^{\hat 1 \hat 2} )^2
     \big)
     +
     \left(
      \frac{f'}{r} - \frac{f}{r^2}
     \right)
     \left(
         \hat h^{\hat 1 \hat 1}
         +
          \hat h^{\hat 2 \hat 2}
     \right)^2
     +
     \frac{f}{r^2}
     \backg_{\hat A \hat D} \backg_{\hat B \hat C} \hat h^{\hat A \hat C} \hat h^{\hat B \hat D}
\nn
  \\
   &=&
   f''
     \left[ (  \hat h^{\hat 1 \hat 2})^2
     +\frac{1}{4}
     \left(
        \hat h^{\hat 1 \hat 1}
         -
        \hat h^{\hat 2 \hat 2}
     \right)^2
     \right]
     +
     \left[
       \frac{(3-n) f}{(n-2)r^2}
       +
       \frac{f'}{r}
       -
       \frac{f''}{4}
     \right]
      \left(
        \hat h^{\hat 1 \hat 1}
         +
        \hat h^{\hat 2 \hat 2}
     \right)^2
       +
     \frac{f}{r^2} |\,\widehat{\hat{h}} \,|^2_{\backg}
    \,.\label{4III18.2}
\end{eqnarray}
\end{widetext}
%
\subsection{Covariant derivatives}
\label{ss7VII19}
We have
 \ptcheck{27V19, PTC, HB and MH together}
\begin{equation}
 \bcov_{\hat{1}} \hat h_{\hat j \hat{k}}
 =
  \sqrt f \partial_r \hat h_{\hat{j} \hat{k}}
  \,.
      \label{24XI18.4}
\end{equation}
Further,
 \ptcheck{27V19, PTC, HB and MH together}
\begin{equation}
 \bcov_{\hat{2}} \hat h_{\hat j \hat{k}}
 =
  \frac{1}{\sqrt f} \partial_\theta \hat h_{\hat{j} \hat{k}}
  -
  \tensor{{\bar \omega}}{^{\hat{\ell}}_{\hat{j}\hat{2}}} \hat h_{\hat{\ell} \hat{k}}
  -
   \tensor{{\bar \omega}}{^{\hat{\ell}}_{\hat{k}\hat{2}}} \hat h_{\hat{\ell} \hat{j}}
  \,.
\end{equation}
The only nonvanishing connection coefficient relevant for this equation is, up to index symmetries, $ \tensor{{\bar \omega}}{^{\hat{1}}_{\hat{2}\hat{2}}} = - \frac{f'}{2 \sqrt{f}}$. This yields
 \ptcheck{27V19, PTC, HB and MH together}
\begin{eqnarray}
 \bcov_{\hat{2}} \hat h_{\hat 1 \hat{1}}
 &=&
  \frac{1}{\sqrt f}
   \left(
      \partial_\theta \hat h_{\hat{1} \hat{1}}
  -
       f'  \hat h_{\hat{1} \hat{2}}
   \right)
 \,,
 \nn
 \\
  \bcov_{\hat{2}} \hat h_{\hat 1 \hat{2}}
 &=&
  \frac{1}{\sqrt f}
   \left[
      \partial_\theta \hat h_{\hat{1} \hat{2}}
   -
       \frac{f'}{2} \left(  \hat h_{\hat{ 2} \hat{2}} -  \hat h_{\hat{1} \hat{1}}  \right)
    \right]
 \,,
 \nn
 \\
  \bcov_{\hat{2}} \hat h_{\hat 2 \hat{2}}
 &=&
  \frac{1}{\sqrt f}
   \left(
     \partial_\theta \hat h_{\hat{2} \hat{2}}
  +
     f'  \hat h_{\hat{1} \hat{2}}
  \right)
  \,,
  \nn
  \\
   \bcov_{\hat{2}} \hat h_{\hat 1 \hat{A}}
 &=&
  \frac{1}{\sqrt f}
   \left(
     \partial_\theta \hat h_{\hat{1} \hat{A}}
  -
     \frac{f'}{2}   \hat h_{\hat{2} \hat{A}}
  \right)
  \,,
  \nn
  \\
    \bcov_{\hat{2}} \hat h_{\hat 2 \hat{A}}
 &=&
  \frac{1}{\sqrt f}
   \left(
     \partial_\theta \hat h_{\hat{2} \hat{A}}
  +
     \frac{f'}{2}   \hat h_{\hat{1} \hat{A}}
  \right)
  \,,
  \nn
  \\
    \bcov_{\hat{2}} \hat h_{\hat A \hat{B}}
 &=&
  \frac{1}{\sqrt f}
     \partial_\theta \hat h_{\hat{A} \hat{B}}
     \,.
      \label{24XI18.2}
\end{eqnarray}
Finally, we have
\begin{equation}
  \bcov_{\hat{A}} \hat h_{\hat j \hat{k}}
 =
  \frac{\ell}{r} \partial_{\hat{A}} \hat h_{\hat{j} \hat{k}}
  -
  \tensor{{\bar \omega}}{^{\hat{\ell}}_{\hat{j}\hat{A}}} \hat h_{\hat{\ell} \hat{k}}
  -
   \tensor{{\bar \omega}}{^{\hat{\ell}}_{\hat{k}\hat{A}}} \hat h_{\hat{\ell} \hat{j}}
  \,.
\end{equation}
In this case the only relevant nonvanishing components of the connection coefficient are, again up to index-symmetries, $ \tensor{{\bar \omega}}{^{\hat{A}}_{\hat{1}\hat{B}}} = \frac{\sqrt{f}}{r} \delta^{\hat{A}}_{\hat{B}} $ {and $\tensor{{\bar \omega}}{^{\hat{1}}_{\hat{A}\hat{B}}} = - \frac{\sqrt{f}}{r} \backg_{\hat{A}\hat{B}}$}  which implies
 \ptcheck{27V19, PTC, HB and MH together}
\begin{eqnarray}
  \bcov_{\hat{A}} \hat h_{\hat 1 \hat{1}}
 &=&
  \frac{\ell}{r} \partial_{\hat{A}} \hat h_{\hat{1} \hat{1}}
  -
  \frac{2 \sqrt{f}}{r} \hat h_{\hat{A} \hat{1}}
  \,,
  \nn
  \\
  \bcov_{\hat{A}} \hat h_{\hat 1 \hat{2}}
 &=&
  \frac{\ell}{r} \partial_{\hat{A}} \hat h_{\hat{1} \hat{2}}
  -
  \frac{\sqrt{f}}{r} \hat h_{\hat{A} \hat{2}}
  \,,
  \nn
  \\
  \bcov_{\hat{A}} \hat h_{\hat 2 \hat{2}}
 &=&
  \frac{\ell}{r} \partial_{\hat{A}} \hat h_{\hat{2} \hat{2}}
  \,,
  \nn
  \\
 \bcov_{\hat{A}} \hat h_{\hat 1 \hat{B}}
  &=&
  \frac{\ell}{r} \partial_{\hat{A}} \hat h_{\hat{1} \hat{B}}
  -
  \frac{\sqrt{f}}{r}
    \left(
       \hat h_{\hat{A} \hat{B}}
       -
       \hat h_{\hat{1} \hat{1}}  \backg_{\hat{A} \hat{B}}
   \right)
  \,,
  \nn
  \\
    \bcov_{\hat{A}} \hat h_{\hat 2 \hat{B}}
 &=&
  \frac{\ell}{r} \partial_{\hat{A}} \hat h_{\hat{2} \hat{B}}
  \,,
  \nn
  \\
    \bcov_{\hat{A}} \hat h_{\hat B \hat{C}}
 &=&
  \frac{\ell}{r} \partial_{\hat{A}} \hat h_{\hat{B} \hat{C}}
  \,.
      \label{24XI18.3}
\end{eqnarray}

\section{The lower bound is not positive for some variations}
 \label{ss18VI19.1}

Recall that
$$
 \xi := \hat{h}_{\hat 1 \hat 1}
         +
     \hat{h}_{\hat 2 \hat 2}
     \,.
$$
Consider that part of the contribution of $\xi$ to \eq{20XI18.10} which decouples from $\phi$:
\begin{eqnarray}
 \int_M \Big\{
 \quadratic{
   \frac{f}{8}  |\partial_r \xi|^2
     -\frac 12
     \left[
       \frac{(3-n) f}{(n-2)r^2}
       +
       \frac{f'}{r}
       -
       \frac{f''}{4}
     \right]
      \xi^2
      }
 \Big\}
 V
  \sqrtbg
  \,.
  \nonumber
  \\
   \label{20XI18.11}
\end{eqnarray}
Note that
\begin{widetext}
$$
-\frac 12
     \left[
       \frac{(3-n) f}{(n-2)r^2}
       +
       \frac{f'}{r}
       -
       \frac{f''}{4}
     \right]
     =
-\frac{ 2 n - (n (n (n+3)-6)-12) \frac{ r_0^n}{r^{n+2} }  }{8
   \ell ^2 (n-2)}
    \,.
   $$
\end{widetext}
%

We show below that, in any dimension, there exist functions which render this integral negative. Hence one cannot neglect the constraint equations and/or the gauge condition when attempting to prove positivity of the  mass for small perturbations of the metric.

If $\xi$ is supported near infinity, the radial part of the integral \eq{20XI18.11} can be written as
\begin{eqnarray}
 \int_{r_0}^\infty \Big\{
 \quadratic{
   \frac{1}{8}  |\partial_r \xi|^2
     -\frac {n}{4(n-2)r^2}
     \xi^2
    + \textrm{l.o.t. } }
 \Big\}
 \frac{r^{n+1}}{\ell^{n+1}}\,dr
  \,,
  \notag
  \\
   \label{20XI18.12}
\end{eqnarray}
where l.o.t. denotes terms which can be made arbitrarily small in comparison with the remaining ones for $r$'s  large enough.
In view of the sharp constant for Hardy's inequality in $(n+2)$-dimensions,
\[
\int_0^\infty |\xi'|^2 \, r^{n+1}\,dr \geq \frac{n^2}{4} \int_0^\infty |\xi|^2\,r^{n-1}\,dr,
\]
 the main terms in \eqref{20XI18.12} are non-negative for all $\xi$ if and only if
\begin{equation}\label{22VI19.2}
\frac{2n}{n-2} \leq \frac{n^2}{4} \Longleftrightarrow n \geq 4
 \,.
\end{equation}
This implies that when $n = 3$,  the integral is negative for open families of functions $\xi$ which are supported sufficiently far away from the origin.

In higher dimensions some more work is needed. In view of \eqref{20XI18.3}, we have
\begin{widetext}
\begin{eqnarray*}
\int_{r_0}^\infty
   f |\partial_r \xi|^2
     r^{n-1}
     \, dr
   &=&
    \int_{r_0}^\infty
   \Big[
    f \big|\partial_r \xi  + \frac{\alpha}{r} \xi\big|^2\
    +
   \Big(\frac{\alpha}{r} \partial_r f + \frac{(n-2)\alpha - \alpha^2}{r^2}  f\Big) |\xi|^2
     \Big]
     r^{n-1}
     \, dr
      \,,
\end{eqnarray*}
and so the integral in \eqref{20XI18.11} becomes
\begin{eqnarray}
 \frac{1}{8}\int_{r_0}^\infty
    f \big|\partial_r \xi  + \frac{\alpha}{r} \xi\big|^2 \frac{r^{n-1}}{\ell^{n-1}}
     \, dr
    +
\frac{1}{8}  \int_{r_0}^\infty \hat W_\alpha(f) |\xi|^2
     \frac{r^{n-1}}{\ell^{n-1}}
     \, dr
      \,,
   \label{20XI18.13}
\end{eqnarray}
where\ptcheck{2 VII 19, by mathematica with MM}
\begin{align*}
\hat W_\alpha(f)
	&= \frac{\alpha - 4}{r}  f'
	+ \frac{(n-2)\alpha - \alpha^2 + \frac{4(n-3)}{n-2}}{r^2}  f
       +
       f''
       \\
       	& =  - \frac{1}{\ell^2} (\alpha^2 - n\alpha + \frac{2n}{n-2})\Big[1 -  \Big(\frac{r_0}{r}\Big)^n \Big] - \frac{1}{\ell^2} n(n+1-\alpha) \Big(\frac{r_0}{r}\Big)^n
	\;.
\end{align*}\ptcheck{2 VII 19, by mathematica with MM}
\end{widetext}
In particular if we choose $\xi = r ^{-\alpha}$ with, when $n \geq 4$,

\begin{equation}\label{22VI19.3}
\frac{1}{2}\Big(n + \sqrt{n^2 - \frac{8n}{n-2}}\Big) < \alpha < n + 1
 \,,
\end{equation}\ptcheck{2 VII 19, by mathematica with MM}
then $\hat W_\alpha(f) < 0$. Therefore, for $\xi = r^{-\alpha}$ with $\alpha$ satisfying \eqref{22VI19.3}, the integral in \eqref{20XI18.13} is negative.


\section{Regularity at $r=r_0$}
 \label{s29V19.1}

We normalise the metric and $r$ so that $r_0 =1=\ell$.

For completeness we verify that metric perturbations of the form
\begin{equation}\label{29V19.7}
  h = h_{\hat 1 \hat 1} (\theta^{\hat 1})^2
  +h_{\hat 2 \hat 2} (\theta^{\hat 2})^2
  +h_{\hat A \hat B} \theta^{\hat A}\theta^{\hat B}
  \,,
\end{equation}
with the $\theta^{\hat a}$'s defined in \eqref{27II18.1}, where  the functions $h_{\hat i \hat j}$ are smooth functions  of $r$  satisfying
\begin{equation}\label{29V19.8}
    h_{\hat 1 \hat 1}|_{r=1}=
   h_{\hat 2 \hat 2}|_{r=1}
  \,,
\end{equation}
define a smooth tensor field on the Horowitz-Myers manifold $\R^2 \times \T^{n-2}$, where the last factor denotes an $(n-2)$-dimensional torus. Here $\R^2$ is parameterised by polar coordinates $(\rho,\varphi)$, where $\rho\in[0,\infty)$ and  $\varphi$ is $2\pi$-periodic,   
defined by the equations
\begin{equation}\label{30V19.1}
  (\rho,\varphi) :=\big(
   \frac{2}{\sqrt{n}} \sqrt{r-1}
    , \frac{n}{2} \theta\big)
  \,.
\end{equation}
In particular functions which are smooth in $r$ near $r=1$ are smooth functions of $\rho^2$ near $\rho=0$, as necessary for smooth rotation-invariant functions. It follows from the last equation that we have
$$
 (\theta^{\hat 2})^2 = { f}{d\theta^2} =
  \big(1+\rho^2 f_1(\rho^2) \big)
   \rho^2 d \varphi^2
 \,,
$$
for some function $f_1$ which is smooth in its argument near $0$. Now,
\begin{eqnarray}
 \nonumber
 \lefteqn{
   h_{\hat 1 \hat 1} (\theta^{\hat 1})^2
  +
   h_{\hat 2 \hat 2} (\theta^{\hat 2})^2
  }
  &&
\\
 \nonumber
 &
  = &
  h_{\hat 1 \hat 1} \big(
    (\theta^{\hat 1})^2 +
  (\theta^{\hat 2})^2
  \big)
  +
  (  h_{\hat 2 \hat 2}
  -   h_{\hat 1 \hat 1})
   (\theta^{\hat 2})^2
\\
   & = &
    h_{\hat 1 \hat 1}  \big(
    (\theta^{\hat 1})^2 +
  (\theta^{\hat 2})^2
  \big) +
   f_2(\rho^2) \rho^4 d \varphi^2
   \,,
   \label{30V19.11}
\end{eqnarray}
for some function $f_2$ which is smooth in its argument near zero. 
\redad{30VIII19}{
The tensor field $ 
    (\theta^{\hat 1})^2 +
  (\theta^{\hat 2})^2
  $ is smooth. In}  Cartesian coordinates we have $\rho^4 d\varphi^2 = (xdy - y dx)^2$ and $\rho^2 = x^2 + y^2$, and smoothness of $  h$ readily follows.

\begin{widetext}
\section{Gauge conditions}
 \label{A29V19.1}

Written out in detail, the harmonicity conditions,
\ptcheck{27V19}
\begin{eqnarray}
 \bcov^i h_{ij} = \frac 12 \bcov_j \phi
  \quad
   \Longleftrightarrow
    \quad
 \bcov^i \hat h_{ij} = \frac {n-2}{2n} \bcov_j \phi
  \,,
\label{29IV19.1c}
\end{eqnarray}
read
\begin{eqnarray}
  \bcov_1
   \big(\hat h_{11} - \frac {n-2}{2n}  \phi \backg_{11}\big) &=&
    - \bcov_2\hat h_{21} - \bcov^A\hat h_{A1}
     \,,
\\
  \bcov_1
    \hat h_{12} &=&
  \frac {n-2}{2n} \bcov_2 \phi
    - \bcov_2\hat h_{22} - \bcov^A\hat h_{A2}
     \,,
\\
  \bcov_1
    \hat h_{1B} &=&
  \frac {n-2}{2n} \bcov_B \phi
    - \bcov_2\hat h_{2B} - \bcov^A\hat h_{AB}
     \,.
\end{eqnarray}
Equivalently, using  
\eqref{24XI18.4}-\eqref{24XI18.3},
%
\begin{eqnarray}
  \sqrt f \partial_r
   \big(\hat h_{\hat 1 \hat 1} - \frac {n-2}{2n}  \phi\big)
   &&=
    - \frac 1 {\sqrt f} \partial_\theta \hat h_{\hat 1\hat 2} - \frac{\ell}{r} \partial_A \hat h_{\hat 1 \hat A}
+ \frac{f'}{2 \sqrt f}
 (\hat h_{\hat 2 \hat 2} - \hat h_{\hat 1 \hat 1})
 -
    \frac{\sqrt f}{r }
    (
    \xi +
 { (n-2)
 \hat h_{\hat 1 \hat 1} )
 }
     \,,
      \label{25II19.3}
\\
  \sqrt f \partial_r
    \hat h_{\hat 1 \hat 2}
    &&=
  \frac {n-2}{2n} \frac 1 {\sqrt f} \partial_\theta\phi
    -  \frac 1 {\sqrt f} \partial_\theta \hat h_{22} - \frac{\ell}{r} \partial_A \hat h_{\hat 2 \hat A }
-
 \frac{f'}{  \sqrt f}
  \hat h_{\hat 1 \hat  2}
     \,,
\\
  \sqrt f \partial_r
    \hat h_{\hat 1 \hat B}
     &&=
  \frac {n-2}{2n} \frac{\ell}{r} \partial_B \phi
    -  \frac 1 {\sqrt f} \partial_\theta \hat h_{\hat 2\hat B} - \frac{\ell}{r} \partial_A \hat h_{\hat A \hat B}
- \frac{f'}{ 2 \sqrt f}
  \hat h_{\hat 1 \hat B}
     \,.
\end{eqnarray}
\end{widetext}
%

\bibliographystyle{aipsamp}
 \ptc{to the PRD typesetters: the PRD bibtex style does not work on our site, using amsplain instead}
\bibliographystyle{amsplain}


\bibliography{BCHMN-minimal}

\end{document}